\pdfoutput=1
\documentclass[preprint,12pt]{elsarticle}
\usepackage{amsmath,amssymb,amsthm}
\usepackage{graphicx}
\usepackage{booktabs}
\usepackage{float}
\usepackage{xcolor}
\usepackage[colorlinks=true,allcolors=blue]{hyperref}
\biboptions{sort&compress}
\newtheorem{proposition}{Proposition}

\providecommand{\AFoldBestGONTHDEK}{3.15}

\providecommand{\BBNLnBGONTHDEK}{-4.66}

\providecommand{\BBNLnBRICCIK}{-2.94}
\providecommand{\BBNLnBRicciMinusGO}{+1.71\pm0.35}

\providecommand{\BBNOdmRicci}{-0.022^{+0.021}_{-0.022}}

\providecommand{\ChiGOHDE}{1485.01}
\providecommand{\ChiGOHDEK}{1484.86}
\providecommand{\ChiGONTHDE}{1484.87}
\providecommand{\ChiGONTHDEK}{1484.79}
\providecommand{\ChiLCDM}{1487.61}
\providecommand{\ChiOCPL}{1484.80}
\providecommand{\ChiOLCDM}{1486.96}
\providecommand{\ChiOWCDM}{1484.92}
\providecommand{\ChiRICCIK}{1484.81}
\providecommand{\ChiTTWOK}{1484.78}
\providecommand{\ChiWCDM}{1485.01}
\providecommand{\CorrOkHGONTHDEK}{-0.03}
\providecommand{\CorrOkWGONTHDEK}{-0.37}
\providecommand{\CorrUWGONTHDEK}{-0.85}
\providecommand{\DAICGOHDE}{+1.40}
\providecommand{\DAICGOHDEK}{+3.25}
\providecommand{\DAICGONTHDE}{+3.26}
\providecommand{\DAICGONTHDEK}{+5.18}
\providecommand{\DAICLCDM}{+0.00}
\providecommand{\DAICOCPL}{+3.19}
\providecommand{\DAICOLCDM}{+1.35}
\providecommand{\DAICOWCDM}{+1.31}
\providecommand{\DAICRICCIK}{+3.20}
\providecommand{\DAICTTWOK}{+5.17}
\providecommand{\DAICWCDM}{-0.60}
\providecommand{\DBICGOHDE}{+12.28}
\providecommand{\DBICGOHDEK}{+19.57}
\providecommand{\DBICGONTHDE}{+19.58}
\providecommand{\DBICGONTHDEK}{+26.93}
\providecommand{\DBICLCDM}{+0.00}
\providecommand{\DBICOCPL}{+19.51}
\providecommand{\DBICOLCDM}{+6.79}
\providecommand{\DBICOWCDM}{+12.19}
\providecommand{\DBICRICCIK}{+19.52}
\providecommand{\DBICTTWOK}{+26.93}
\providecommand{\DBICWCDM}{+4.84}
\providecommand{\DChiGOHDE}{-2.60}
\providecommand{\DChiGOHDEK}{-2.75}
\providecommand{\DChiGONTHDE}{-2.74}
\providecommand{\DChiGONTHDEK}{-2.82}
\providecommand{\DChiLCDM}{+0.00}
\providecommand{\DChiOCPL}{-2.81}
\providecommand{\DChiOLCDM}{-0.65}
\providecommand{\DChiOWCDM}{-2.69}
\providecommand{\DChiRICCIK}{-2.80}
\providecommand{\DChiTTWOK}{-2.83}
\providecommand{\DChiWCDM}{-2.60}
\providecommand{\DDICGOHDE}{-1.33}
\providecommand{\DDICGOHDEK}{+0.47}
\providecommand{\DDICGONTHDE}{-0.91}
\providecommand{\DDICGONTHDEK}{+0.81}
\providecommand{\DDICLCDM}{+0.00}
\providecommand{\DDICOCPL}{+2.96}
\providecommand{\DDICOLCDM}{+0.72}
\providecommand{\DDICOWCDM}{+0.90}
\providecommand{\DDICRICCIK}{+1.36}
\providecommand{\DDICTTWOK}{+1.81}
\providecommand{\DDICWCDM}{-1.15}

\providecommand{\DerAFoldGONTHDEK}{1.228^{+1.785}_{-0.210}}

\providecommand{\DerDeltaGONTHDEK}{2.000^{+0.312}_{-0.830}}

\providecommand{\DerOdkGONTHDEK}{0.316^{+0.101}_{-0.084}}

\providecommand{\DerOktGONTHDEK}{-0.003^{+0.060}_{-0.068}}

\providecommand{\DerQzGONTHDEK}{-0.432^{+0.125}_{-0.074}}

\providecommand{\DerRkGONTHDEK}{1.461^{+0.252}_{-0.160}}

\providecommand{\DerTFoldGONTHDEK}{2.9^{+15.7}_{-2.7}}

\providecommand{\DerWaGONTHDEK}{-0.278^{+0.369}_{-2.915}}

\providecommand{\DerWzGONTHDEK}{-0.834^{+0.105}_{-0.086}}

\providecommand{\DerZtGONTHDEK}{0.694^{+0.034}_{-0.032}}

\providecommand{\HCBGONTHDEK}{66.65\pm2.17}

\providecommand{\HSNGONTHDEK}{73.36\pm1.04}

\providecommand{\JHEAPOk}{0.012^{+0.015}_{-0.025}}
\providecommand{\KGOHDE}{5}
\providecommand{\KGOHDEK}{6}
\providecommand{\KGONTHDE}{6}
\providecommand{\KGONTHDEK}{7}
\providecommand{\KLCDM}{3}
\providecommand{\KOCPL}{6}
\providecommand{\KOLCDM}{4}
\providecommand{\KOWCDM}{5}
\providecommand{\KRICCIK}{6}
\providecommand{\KTTWOK}{7}
\providecommand{\KWCDM}{4}
\providecommand{\LitHalPPDAIC}{+2.07}
\providecommand{\LitHalPPDChi}{+0.07}
\providecommand{\LitHalPPLnZ}{-3.73}
\providecommand{\LitHalUDAIC}{+1.73}
\providecommand{\LitHalUDChi}{-0.27}
\providecommand{\LitHalULnZ}{-2.19}
\providecommand{\LitLPSDAIC}{+4.2}
\providecommand{\LitLPSDChi}{+0.2}
\providecommand{\LitLucPPDAIC}{+1.73}
\providecommand{\LitLucPPDChi}{-4.27}
\providecommand{\LitLucPPDelta}{-0.041}
\providecommand{\LitLucUDAIC}{-3.63}
\providecommand{\LitLucUDChi}{-9.63}
\providecommand{\LitLucUDelta}{-0.234}
\providecommand{\LitMahDAIC}{-4.03}
\providecommand{\LitMahDChi}{-10.03}
\providecommand{\LitYadOk}{0.012^{+0.015}_{-0.025}}

\providecommand{\LitYadW}{-0.922\pm0.047}
\providecommand{\LnBAGOHDE}{-2.33}
\providecommand{\LnBAGOHDEK}{-3.51}
\providecommand{\LnBAGONTHDE}{-1.28}
\providecommand{\LnBAGONTHDEK}{-2.61}
\providecommand{\LnBALCDM}{+0.00}
\providecommand{\LnBAOCPL}{-5.54}
\providecommand{\LnBAOLCDM}{-1.50}
\providecommand{\LnBAOWCDM}{-3.36}
\providecommand{\LnBARICCIK}{-2.94}
\providecommand{\LnBATTWOK}{-2.66}
\providecommand{\LnBAWCDM}{-1.87}
\providecommand{\LnBCurvVsFlat}{-1.39\pm0.25}
\providecommand{\LnBCurvVsFlatHDE}{-1.28\pm0.27}
\providecommand{\LnBGOHDE}{-2.33\pm0.24}
\providecommand{\LnBGOHDEK}{-3.61\pm0.25}
\providecommand{\LnBGOHDEvsOWCDM}{-0.25\pm0.27}
\providecommand{\LnBGONTHDE}{-1.28\pm0.24}
\providecommand{\LnBGONTHDEK}{-2.66\pm0.24}
\providecommand{\LnBGOvsRicci}{+0.28\pm0.26}
\providecommand{\LnBLCDM}{+0.00\pm0.22}
\providecommand{\LnBNTHDEvsHDEK}{+0.94\pm0.27}
\providecommand{\LnBOCPL}{-5.79\pm0.26}
\providecommand{\LnBOLCDM}{-1.50\pm0.23}
\providecommand{\LnBOWCDM}{-3.36\pm0.25}
\providecommand{\LnBRICCIK}{-2.94\pm0.24}
\providecommand{\LnBTTWOK}{-2.70\pm0.24}
\providecommand{\LnBTwin}{-0.46\pm0.26}
\providecommand{\LnBWCDM}{-1.87\pm0.24}

\providecommand{\LnBvsOLCDMGONTHDEK}{-1.16\pm0.25}

\providecommand{\MockNin}{6}
\providecommand{\MockNpar}{7}
\providecommand{\MockOdmHi}{0.226}
\providecommand{\MockOdmTruth}{0.229}
\providecommand{\MockOkPost}{0.049^{+0.041}_{-0.040}}

\providecommand{\MockOkTruth}{0.050}
\providecommand{\MockUzeroSd}{0.41}
\providecommand{\MockUzeroTruth}{0.80}
\providecommand{\NEqChecks}{25}
\providecommand{\NEqOK}{25}
\providecommand{\NTests}{17}

\providecommand{\OkHiGONTHDEK}{+0.080}

\providecommand{\OkLoGONTHDEK}{-0.087}

\providecommand{\PFoldGONTHDE}{0.94}
\providecommand{\PFoldGONTHDEK}{0.91}

\providecommand{\POpenGONTHDEK}{0.49}

\providecommand{\POpenOLCDM}{0.79}

\providecommand{\PPhantomGONTHDEK}{0.02}

\providecommand{\PostHzeroGOHDEK}{72.00^{+0.87}_{-0.86}}

\providecommand{\PostHzeroGONTHDEK}{71.88^{+0.90}_{-0.90}}
\providecommand{\PostHzeroLCDM}{72.34^{+0.85}_{-0.84}}
\providecommand{\PostHzeroOCPL}{71.93^{+0.90}_{-0.86}}
\providecommand{\PostHzeroOLCDM}{72.21^{+0.85}_{-0.85}}
\providecommand{\PostHzeroOWCDM}{72.00^{+0.87}_{-0.86}}
\providecommand{\PostHzeroRICCIK}{71.85^{+0.87}_{-0.87}}
\providecommand{\PostHzeroTTWOK}{71.91^{+0.90}_{-0.88}}

\providecommand{\PostOdmGONTHDEK}{0.129^{+0.104}_{-0.087}}

\providecommand{\PostOkGOHDEK}{-0.012^{+0.045}_{-0.045}}

\providecommand{\PostOkGONTHDEK}{-0.001^{+0.041}_{-0.043}}

\providecommand{\PostOkOCPL}{-0.005^{+0.048}_{-0.047}}
\providecommand{\PostOkOLCDM}{0.029^{+0.036}_{-0.035}}
\providecommand{\PostOkOWCDM}{-0.009^{+0.046}_{-0.045}}
\providecommand{\PostOkRICCIK}{0.002^{+0.043}_{-0.044}}
\providecommand{\PostOkTTWOK}{-0.004^{+0.041}_{-0.042}}

\providecommand{\PostOmLCDM}{0.301^{+0.008}_{-0.008}}
\providecommand{\PostOmOCPL}{0.299^{+0.017}_{-0.023}}
\providecommand{\PostOmOLCDM}{0.294^{+0.011}_{-0.011}}
\providecommand{\PostOmOWCDM}{0.298^{+0.012}_{-0.012}}

\providecommand{\PostOmtGOHDEK}{0.300^{+0.019}_{-0.018}}

\providecommand{\PostOmtGONTHDEK}{0.304^{+0.017}_{-0.017}}

\providecommand{\PostOmtRICCIK}{0.304^{+0.013}_{-0.013}}
\providecommand{\PostOmtTTWOK}{0.304^{+0.017}_{-0.017}}

\providecommand{\PostRdGOHDEK}{139.35^{+1.75}_{-1.74}}

\providecommand{\PostRdGONTHDEK}{139.49^{+1.82}_{-1.75}}
\providecommand{\PostRdLCDM}{139.96^{+1.75}_{-1.73}}
\providecommand{\PostRdOCPL}{139.43^{+1.77}_{-1.74}}
\providecommand{\PostRdOLCDM}{139.87^{+1.74}_{-1.72}}
\providecommand{\PostRdOWCDM}{139.36^{+1.77}_{-1.73}}
\providecommand{\PostRdRICCIK}{139.54^{+1.77}_{-1.73}}
\providecommand{\PostRdTTWOK}{139.47^{+1.77}_{-1.77}}

\providecommand{\PostWOWCDM}{-0.931^{+0.046}_{-0.050}}

\providecommand{\PostWgoGOHDEK}{-0.922^{+0.062}_{-0.070}}

\providecommand{\PostWgoGONTHDEK}{-1.696^{+0.514}_{-0.476}}

\providecommand{\PostWgoRICCIK}{-1.789^{+0.556}_{-0.416}}
\providecommand{\PostWgoTTWOK}{-1.666^{+0.497}_{-0.475}}

\providecommand{\PostWzeroOCPL}{-0.914^{+0.061}_{-0.062}}

\providecommand{\QDMAPGONTHDEK}{9.87}
\providecommand{\QDMAPOLCDM}{9.74}

\providecommand{\QSigGONTHDEK}{2.6}
\providecommand{\QSigOLCDM}{2.6}

\providecommand{\RobDAICPPNOSHOES}{+3.64}

\providecommand{\RobDAICUNIONTHREE}{-1.35}
\providecommand{\RobDChiBBN}{-3.69}

\providecommand{\RobDChiCCSYS}{-3.73}

\providecommand{\RobDChiDOVEKIE}{-6.48}

\providecommand{\RobDChiODMWIDE}{-2.86}

\providecommand{\RobDChiOKGAUSS}{-2.87}

\providecommand{\RobDChiOKNARROW}{-2.83}

\providecommand{\RobDChiOKWIDE}{-2.83}

\providecommand{\RobDChiPPNOSHOES}{-4.36}

\providecommand{\RobDChiUHALF}{-2.83}

\providecommand{\RobDChiUNIONTHREE}{-9.35}

\providecommand{\RobLnBABBN}{-4.62}
\providecommand{\RobLnBACCSYS}{-2.23}
\providecommand{\RobLnBADOVEKIE}{-0.95}
\providecommand{\RobLnBAODMWIDE}{-2.99}
\providecommand{\RobLnBAOKGAUSS}{-6.00}
\providecommand{\RobLnBAOKNARROW}{-1.78}
\providecommand{\RobLnBAOKWIDE}{-3.39}
\providecommand{\RobLnBAPPNOSHOES}{-1.89}
\providecommand{\RobLnBAUHALF}{-2.90}
\providecommand{\RobLnBAUNIONTHREE}{+0.58}
\providecommand{\RobLnBBBN}{-4.66}

\providecommand{\RobLnBCCSYS}{-2.29}

\providecommand{\RobLnBDOVEKIE}{-1.00}

\providecommand{\RobLnBODMWIDE}{-3.03}

\providecommand{\RobLnBOKGAUSS}{-6.05}

\providecommand{\RobLnBOKNARROW}{-1.79}

\providecommand{\RobLnBOKWIDE}{-3.58}

\providecommand{\RobLnBPPNOSHOES}{-1.94}

\providecommand{\RobLnBUHALF}{-2.97}

\providecommand{\RobLnBUNIONTHREE}{+0.53}

\providecommand{\RobOkBBN}{+0.001^{+0.040}_{-0.042}}
\providecommand{\RobOkCCSYS}{+0.008^{+0.041}_{-0.042}}
\providecommand{\RobOkDOVEKIE}{+0.006^{+0.042}_{-0.038}}
\providecommand{\RobOkODMWIDE}{-0.007^{+0.038}_{-0.039}}
\providecommand{\RobOkOKGAUSS}{+0.001^{+0.002}_{-0.002}}
\providecommand{\RobOkOKNARROW}{+0.004^{+0.039}_{-0.041}}
\providecommand{\RobOkOKWIDE}{-0.001^{+0.044}_{-0.044}}
\providecommand{\RobOkPPNOSHOES}{+0.017^{+0.042}_{-0.046}}
\providecommand{\RobOkUHALF}{-0.003^{+0.043}_{-0.045}}
\providecommand{\RobOkUNIONTHREE}{+0.001^{+0.045}_{-0.048}}
\providecommand{\RobRkBBN}{1.42^{+0.25}_{-0.10}}

\providecommand{\RobRkODMWIDE}{1.61^{+0.53}_{-0.25}}

\providecommand{\TFoldBestGONTHDEK}{20.1}

\providecommand{\UzeroULGONTHDEK}{1.52}

\providecommand{\UzeroULRICCIK}{1.52}
\providecommand{\UzeroULTTWOK}{1.50}

\newcommand{\Omt}{\tilde\Omega_m}
\newcommand{\Okt}{\tilde\Omega_k}
\newcommand{\wgo}{w_{\rm GO}}
\newcommand{\Odm}{\Omega_d^{(m)}}
\newcommand{\Odr}{\Omega_d^{(r)}}
\newcommand{\Odk}{\Omega_d^{(k)}}
\newcommand{\Rk}{R_k}
\newcommand{\Ok}{\Omega_{k0}}
\newcommand{\keff}{k_{\rm eff}}
\newcommand{\Mp}{M_{\rm p}}
\newcommand{\SBH}{S_{\rm BH}}
\newcommand{\pipefig}[2]{\IfFileExists{figures/#1.pdf}{\includegraphics[width=#2]{figures/#1}}%
  {\fbox{\parbox[c][0.16\textheight][c]{\dimexpr #2-2\fboxsep-2\fboxrule\relax}{\centering\footnotesize figure produced by the pipeline\\(\texttt{\detokenize{figures/#1}})}}}}

\journal{Nuclear Physics B}
\makeatletter
\def\ps@pprintTitle{%
  \let\@oddhead\@empty
  \let\@evenhead\@empty
  \let\@oddfoot\@empty
  \let\@evenfoot\@oddfoot}
\makeatother

\begin{document}
\begin{frontmatter}

\title{New Tsallis holographic dark energy with the Granda--Oliveros cutoff in a non-flat universe: curvature tracking and DESI DR2 constraints}

\author[gla]{Umesh Kumar Sharma}
\ead{sharma.umesh@gla.ac.in}
\author[iilm]{Pankaj}
\ead{pankaj.fellow@yahoo.co.in}
\author[utas]{Bramha Dutta Pandey}
\ead{bdpandey05@gmail.com}
\author[utas]{P.~Suresh~Kumar\corref{cor1}}
\ead{sureshharinirahav@gmail.com}
\cortext[cor1]{Corresponding author}

\affiliation[gla]{organization={Department of Mathematics, Institute of Applied Sciences and Humanities, GLA University},
  city={Mathura}, postcode={281406}, state={Uttar Pradesh}, country={India}}
\affiliation[iilm]{organization={School of Sciences, IILM University, Knowledge Park-II},
  city={Greater Noida}, postcode={201310}, country={India}}
\affiliation[utas]{organization={Department of Supportive Requirements, University of Technology and Applied Sciences},
  city={Salalah}, country={Sultanate of Oman}}

\begin{abstract}
Holographic dark energy ties the dark-energy density to an infrared length. We study new Tsallis holographic dark energy (NTHDE), built on the exponential Tsallis entropy, with the local Granda--Oliveros (GO) cutoff $L^{-2}=\alpha H^2+\gamma\dot H$ in a spatially curved universe. Writing the background as an autonomous system for $(\Omega_d,u,\Omega_k)$, we obtain four exact results. (i) In the closed-form solution without the Tsallis deformation, the dark energy splits into a matter-like and a curvature-like part. (ii) Because the GO cutoff lacks the term $k/a^2$ of the Ricci scalar, the coefficient of $(1+z)^2$ in $H^2/H_0^2$ is $\Rk\,\Omega_{k0}$ with $\Rk=(1-\alpha+\gamma)^{-1}>1$ at every physical point, while distances and the Clarkson--Bassett--Lu test return $\Omega_{k0}$; the Ricci completion restores $\Rk=1$. (iii) Big-bang nucleosynthesis locks the matter-era dark-energy fraction to $\Omega_d^{(m)}=1/(1+3|\wgo|)$, where $\wgo$ is the equation of state of the GO mode, and then $\Rk=(1+3|\wgo|)/(3|\wgo|-1)$. (iv) The finite-time fold singularity of the flat model survives curvature. Against Pantheon+SH0ES supernovae, 32 cosmic chronometers and DESI DR2 BAO, the curvature of curved GO-NTHDE is consistent with zero, $\Omega_{k0}=\PostOkGONTHDEK$. Curved GO-NTHDE gives $\Delta\chi^2=\DChiGONTHDEK$ relative to flat $\Lambda$CDM, but the Bayesian evidence prefers flat $\Lambda$CDM ($\ln B=\LnBGONTHDEK$, model against $\Lambda$CDM), and allowing curvature costs $\ln B=\LnBCurvVsFlat$ relative to flat GO-NTHDE. The data cannot yet measure $\Rk$: the GO and Ricci cutoffs fit equally well, $\ln B=\LnBGOvsRicci$. With Union3 supernovae the model is on a par with $\Lambda$CDM, $\ln B=\RobLnBUNIONTHREE$. Nucleosynthesis locks the dark-energy fraction onto the Ricci surface as in the flat case, and a fraction $\PFoldGONTHDEK$ of the posterior reaches the finite-time fold.
\end{abstract}

\end{frontmatter}

\section{Introduction}
\label{sec:intro}

The accelerated expansion of the Universe~\cite{SupernovaSearchTeam:1998fmf,SupernovaCosmologyProject:1998vns} is described well by a cosmological constant, whose value nevertheless remains unexplained~\cite{Weinberg:1988cp,Copeland:2006wr}. Baryon acoustic oscillation (BAO) data from DESI DR2, combined with supernovae, show a mild preference for evolving dark energy that depends on the supernova compilation~\cite{DESI:2025zgx,Rubin:2023ovl,Popovic:2025dovekie}. A second question has sharpened at the same time: is space exactly flat? Planck temperature and polarization spectra prefer a closed universe at more than $3\sigma$ in some analyses~\cite{Planck:2018vyg,Handley:2019tkm,DiValentino:2019qzk}, an effect traced to the lensing amplitude and disputed in others~\cite{Efstathiou:2020wem}; DESI DR2 BAO combined with the CMB prefer a small negative curvature at about $2\sigma$ within $\Lambda$CDM~\cite{Chen:2025mlf}. Late-time analyses of dynamical dark energy with curvature give small, model-dependent values of $\Omega_{k0}$ with uncertainties of a few times $10^{-2}$; for $w$CDM with curvature, DESI DR2 with BBN, chronometers and Pantheon+ gives $\Omega_{k0}=\JHEAPOk$~\cite{Yadav:2026jheap}.\footnote{Taken from Table~1 of Ref.~\cite{Yadav:2026jheap}, combination DR2+BBN+OHD+PP; sign convention $\Omega_{k0}>0$ open.} Dark energy and curvature are partly degenerate in background data, so a dark-energy model that carries its own relation to curvature is of interest.

Holographic dark energy (HDE) ties the dark-energy density to an infrared length $L$ through the holographic bound~\cite{tHooft:1993dmi,Susskind:1994vu,Fischler:1998st,Bousso:2002ju,Cohen:1998zx,Li:2004rb,Wang:2016och}. The choice of $L$ decides the dynamics: the Hubble radius cannot accelerate the Universe~\cite{Hsu:2004ri}, and the future event horizon does so but depends on the whole future evolution of the Universe~\cite{Li:2004rb,Huang:2004ai}. Granda and Oliveros (GO) proposed the local cutoff $L^{-2}=\alpha H^2+\gamma\dot H$ (their $\beta$ is our $\gamma$)~\cite{Granda:2008dk,Granda:2008tm}. Karami and Fehri extended it to a curved universe by keeping the same form of the cutoff and adding the curvature density to the Friedmann equation~\cite{Karami:2009gof}. They solved the model without data. Curvature has also been added to Barrow HDE with a GO cutoff in the modified Friedmann equations of Barrow cosmology~\cite{Motaghi:2024gom}, and to Barrow HDE with other cutoffs~\cite{Adhikary:2021xym,Luciano:2022ffn}; in Barrow cosmology the GO cutoff itself was later adapted to the modified Friedmann equation, in a flat universe~\cite{Mahmoudifard:2025pdu}. NTHDE itself was introduced with the future event horizon in a flat universe~\cite{Pandey:2021fvr}.

Generalized horizon entropies modify HDE further: the power-law Tsallis and Barrow forms~\cite{Tavayef:2018xwx,Saridakis:2018unr,Nojiri:2022aof} and the exponential Tsallis entropy behind NTHDE~\cite{Moradpour:2020dfm,Pandey:2021fvr}. Power-law models with the GO cutoff were recently fitted to PantheonPlus, Union3, chronometers and DESI DR2 in a flat universe~\cite{Luciano:2026hints}, as were Barrow and Tsallis HDE with the event horizon~\cite{Luciano:2025jheap} and HDE built on the Viaggiu entropy~\cite{Halder:2026plb}, introduced in Ref.~\cite{Saha:2026npb}; in all of these the information criteria or the evidence find at most a weak preference for the holographic model, and usually none. The flat NTHDE-GO model, solved as an autonomous system, has three exact properties~\cite{Pandey:2026flat}: without the Tsallis deformation its expansion rate is exactly that of flat $w$CDM, so the matter-era dark-energy fraction is a blind direction of late-time data; nucleosynthesis locks that direction; and trajectories with $w_d>-1$ reach a floor of the density in finite time. The holographic bound $\rho L^4\lesssim S$ is a statement about the number of degrees of freedom of a horizon, and the NTHDE density is its realization for the exponential Tsallis entropy; holographic models with other generalized entropies, such as the Viaggiu entropy, have been studied in the same spirit~\cite{Saha:2026npb}.

This paper asks what curvature does to that picture. Its new elements are: (i) the splitting of the curved GO-HDE solution of Karami and Fehri~\cite{Karami:2009gof} into a matter-like and a curvature-like part, and the resulting curvature tracker of the dark energy; (ii) a curvature consistency ratio $\Rk$ that measures how much the expansion rate overstates the magnitude of the spatial curvature because the GO cutoff lacks the curvature term of the Ricci scalar, with the Ricci completion as the exception; (iii) the proof that the BBN lock and the fold singularity of the flat model survive curvature; (iv) a first data analysis of curved GO-NTHDE with explicit prior-volume, tension and robustness checks. Our main data result is that current late-time data are consistent with a flat universe in this model, $\Omega_{k0}=\PostOkGONTHDEK$, and cannot yet measure $\Rk$. Every equation is either taken from the cited literature or derived step by step from earlier numbered equations, so that the paper can be followed without prior experience of holographic dark energy.

The paper is organized as follows. Section~\ref{sec:model} builds the model and derives its equations and exact results. Section~\ref{sec:data} defines the data and likelihoods, Sect.~\ref{sec:stats} the statistical methods, Sect.~\ref{sec:results} presents the results and Sect.~\ref{sec:conclusions} concludes. Table~\ref{tab:symbols} lists the symbols; every figure and table is placed next to the text that first uses it.

\begin{table}[H]
\centering\small
\caption{Symbols used in this paper. Units: $\hbar=k_B=1$; $c_0$ denotes the speed of light where needed. Curvature convention: $\Omega_{k0}=-k/H_0^2$, so $\Omega_{k0}>0$ is an open universe.}
\label{tab:symbols}
\begin{tabular}{p{0.30\textwidth}p{0.64\textwidth}}
\toprule
Symbol & Meaning\\
\midrule
$a$, $z$, $x$ & scale factor ($a_0=1$ today), redshift $1+z=1/a$, e-folding variable $x=\ln a$\\
$H$, $H_0$, $E$ & Hubble rate $\dot a/a$, its present value, $E\equiv H/H_0$; a prime is ${\rm d}/{\rm d}x$\\
$k$, $\Omega_k$, $\Ok$ & spatial curvature constant; $\Omega_k=-k/(a^2H^2)$; value today $-k/H_0^2$\\
$\Mp$, $G$ & reduced Planck mass $\Mp^2=(8\pi G)^{-1}$, Newton's constant\\
$L$ & infrared cutoff length of holographic dark energy\\
$\SBH$, $S_T$ & Bekenstein--Hawking and Tsallis entropies of a horizon of radius $L$\\
$\delta$, $n$ & Tsallis parameter $\delta=1-n$ ($\delta\to0$: standard entropy)\\
$c$ & holographic constant; absorbed into $\alpha,\gamma,b$ (set to 1, Sect.~\ref{sec:dimless})\\
$\alpha$, $\gamma$, $\kappa$ & GO cutoff constants, $L^{-2}=\alpha H^2+\gamma\dot H+\kappa\,k/a^2$; $\kappa=0$: GO; $\kappa=\gamma$, $\alpha=2\gamma$: Ricci\\
$b$, $u$, $u_0$ & dimensionless deformation $b=8\pi^2c^2\delta\Mp^2/H_0^2$; $u=\delta\SBH=b/X$; its value today\\
$g(u)$, $\keff(u)$ & $g=(e^u-1)/u$; entropy slope $\keff={\rm d}\ln S_T/{\rm d}\ln\SBH$\\
$u_*$ & fold of the NTHDE branch, $\keff(u_*)=2$, $u_*=1.5936$\\
$\Omega_{m0}$, $\Omega_d$ & matter density today; dark-energy fraction $\rho_d/(3\Mp^2H^2)$\\
$\Omt$, $\Okt$ & matter-like and curvature-like coefficients of $a^{-3}$ and $a^{-2}$ in $H^2/H_0^2$ (Prop.~\ref{prop:curved})\\
$\Rk$ & curvature consistency ratio $\Okt/\Ok$\\
$\Odm$, $\Odr$, $\Odk$ & dark-energy fraction during matter, radiation and curvature domination\\
$w_d$, $\wgo$, $q$ & dark-energy equation of state, GO-mode equation of state, deceleration parameter\\
$\Delta N_\nu$ & extra effective neutrino species at nucleosynthesis\\
$D_M$, $D_H$, $D_V$, $D_L$, $\mu$ & transverse comoving, Hubble, volume-averaged and luminosity distances; distance modulus\\
$r_d$ & comoving sound horizon at the baryon-drag epoch\\
$\ln B$ & logarithm of the Bayes factor against flat $\Lambda$CDM\\
\bottomrule
\end{tabular}
\end{table}

\section{The model}
\label{sec:model}

\subsection{Holographic dark energy and the NTHDE density}
\label{sec:hde}
Cohen, Kaplan and Nelson argued that an effective field theory in a region of size $L$ should not contain more energy than a black hole of the same size, $L^3\rho\lesssim\Mp^2L$~\cite{Cohen:1998zx}. Saturating this bound with a dimensionless constant $c$ defines holographic dark energy~\cite{Li:2004rb,Wang:2016och},
\begin{equation}
\rho_d=\frac{3c^2\Mp^2}{L^2}.
\label{eq:hde}
\end{equation}
The bound is an entropy statement. A horizon of radius $L$ has area $A=4\pi L^2$ and Bekenstein--Hawking entropy~\cite{PhysRevD.7.2333,Hawking:1975vcx}
\begin{equation}
\SBH=\frac{A}{4G}=\frac{\pi L^2}{G}=8\pi^2\Mp^2L^2,
\label{eq:sbh}
\end{equation}
where the last step uses $G=(8\pi\Mp^2)^{-1}$. Because $\SBH\propto L^2$, Eq.~\eqref{eq:hde} can be written as $\rho_d=(3c^2\Mp^2/L^2)(S/\SBH)$ with $S=\SBH$. Following Refs.~\cite{Moradpour:2020dfm,Pandey:2021fvr}, a generalized entropy $S$ defines a generalized HDE through the same relation, normalized so that Eq.~\eqref{eq:hde} is recovered for $S\to\SBH$:
\begin{equation}
\rho_d=\frac{3c^2\Mp^2}{L^2}\,\frac{S}{\SBH}.
\label{eq:ghde}
\end{equation}
For a system with $W$ microstates of probabilities $P_i$ the Tsallis entropy is~\cite{Tsallis:1987eu} $S_T=(1-n)^{-1}\sum_i(P_i^{\,n}-P_i)$, which tends to the Gibbs--Shannon entropy for $n\to1$. For $W$ equally likely states, $P_i=1/W$, it equals $S_T=(W^{1-n}-1)/(1-n)$. For a horizon, Boltzmann's relation $\SBH=\ln W$ gives $W=e^{\SBH}$~\cite{Tsallis:2012js,Moradpour:2020dfm}, and with $\delta\equiv1-n$
\begin{equation}
S_T=\frac{e^{\delta\SBH}-1}{\delta}=\SBH+\frac{\delta\,\SBH^2}{2}+\frac{\delta^2\SBH^3}{6}+O(\delta^3).
\label{eq:st}
\end{equation}
Substituting Eqs.~\eqref{eq:sbh} and~\eqref{eq:st} in Eq.~\eqref{eq:ghde} gives the NTHDE density~\cite{Moradpour:2020dfm}
\begin{equation}
\begin{gathered}
\rho_d=\frac{3c^2\Mp^2}{L^2}\,g(u),\qquad g(u)=\frac{e^u-1}{u}=1+\frac u2+\frac{u^2}{6}+O(u^3),\\
u\equiv\delta\SBH=8\pi^2\delta\Mp^2L^2 .
\end{gathered}
\label{eq:rho}
\end{equation}
Keeping three terms of the series (we call this truncation T2) reproduces the three-term density $\rho_d=3c^2\Mp^2L^{-2}+\delta_g\Mp^4+\delta_g^2\Mp^6L^2$ of Ref.~\cite{Pandey:2021fvr} for $c^2=2/9$, $\delta_g=8\pi^2\delta/3$: inserting three terms of the series in Eq.~\eqref{eq:rho} gives coefficients $\tfrac32c^2(8\pi^2\delta)$ and $\tfrac12c^2(8\pi^2\delta)^2$, and equating them with $\delta_g$ and $\delta_g^2$ gives $c^2=2/9$. The expansion parameter is $u=\delta\SBH$, not $\delta$: with $\SBH\sim10^{122}$ for the present horizon, $u=O(1)$ corresponds to $\delta\sim10^{-122}$, so the truncation is not a controlled approximation. We use the full $g(u)$ and keep T2 for comparison.

The logarithmic slope of the Tsallis entropy with respect to $\SBH$ will be needed. Since ${\rm d}S_T/{\rm d}\SBH=e^{\delta\SBH}$,
\begin{equation}
\keff(u)\equiv\frac{{\rm d}\ln S_T}{{\rm d}\ln\SBH}=\frac{u\,e^u}{e^u-1}=1+\frac{u\,g'(u)}{g(u)} .
\label{eq:keff}
\end{equation}
Figure~\ref{fig:tsallis} shows $g$ and $\keff$; the points where $\keff=2$ are the folds of Sect.~\ref{sec:formB}.

\begin{figure}[H]
\centering\pipefig{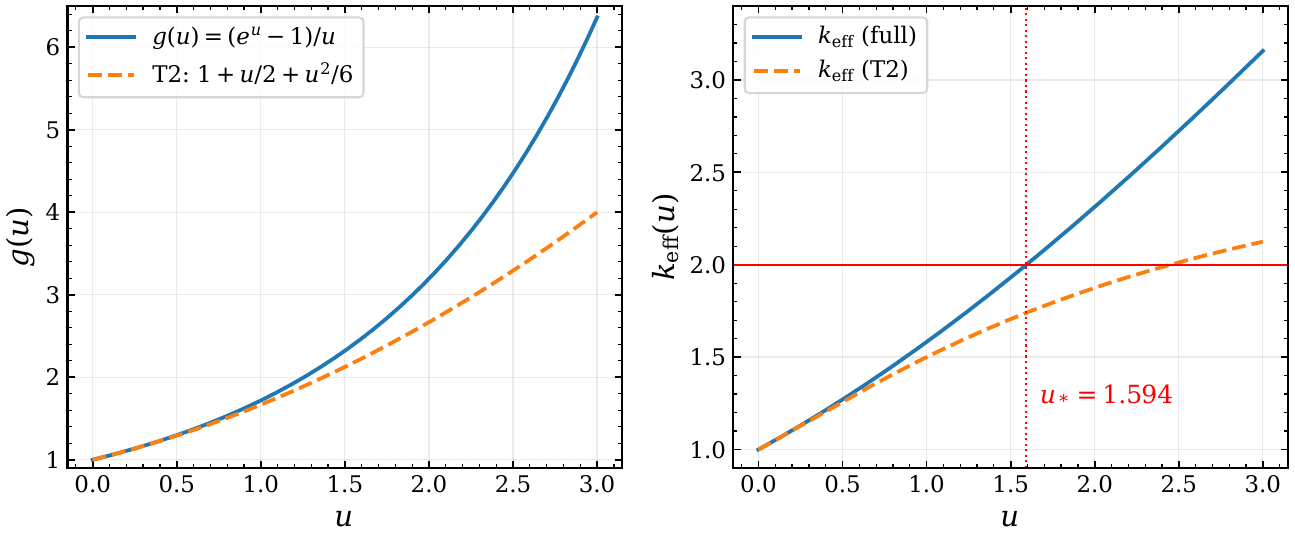}{0.85\textwidth}
\caption{$g(u)$, Eq.~\eqref{eq:rho}, and $\keff(u)$, Eq.~\eqref{eq:keff}, exact and truncated (T2). Vertical dotted line: the fold of the full density, $u_*=1.5936$ [Eq.~\eqref{eq:ustar}]; the T2 fold is at $\sqrt6$. Horizontal line: $\keff=2$.}
\label{fig:tsallis}
\end{figure}

\subsection{Curved universe and the Granda--Oliveros cutoff}
\label{sec:go}
We consider a homogeneous and isotropic universe with spatial curvature constant $k$,
\begin{equation}
{\rm d}s^2=-{\rm d}t^2+a^2(t)\left[\frac{{\rm d}r^2}{1-kr^2}+r^2\,{\rm d}\Omega^2\right],
\label{eq:metric}
\end{equation}
($k=0$ flat, $k>0$ closed, $k<0$ open), filled with pressureless matter and NTHDE. The Friedmann equations are~\cite{Weinberg:2008zzc,baumann2022cosmology}
\begin{align}
3\Mp^2\left(H^2+\frac k{a^2}\right)&=\rho_m+\rho_d, \label{eq:fried1}\\
-2\Mp^2\left(\dot H-\frac k{a^2}\right)&=\rho_m+\rho_d+p_d, \label{eq:fried2}
\end{align}
with $H=\dot a/a$ and $p_d$ the dark-energy pressure. Matter is conserved separately, $\rho_m=\rho_{m0}a^{-3}$. Differentiating Eq.~\eqref{eq:fried1} and using Eq.~\eqref{eq:fried2} gives the conservation law of dark energy,
\begin{equation}
\dot\rho_d+3H(1+w_d)\rho_d=0,\qquad w_d\equiv\frac{p_d}{\rho_d}.
\label{eq:cons}
\end{equation}
We normalize $a(t_0)=1$, use $1+z=1/a$ and $x=\ln a$, write primes for ${\rm d}/{\rm d}x$, and define
\begin{equation}
\begin{gathered}
E\equiv\frac{H}{H_0},\qquad \Omega_{m0}\equiv\frac{\rho_{m0}}{3\Mp^2H_0^2},\qquad \Ok\equiv-\frac{k}{H_0^2},\\
\Omega_d\equiv\frac{\rho_d}{3\Mp^2H^2},\qquad \Omega_k\equiv-\frac{k}{a^2H^2},\qquad q\equiv-1-\frac{\dot H}{H^2}.
\end{gathered}
\label{eq:defs}
\end{equation}
Then $\Omega_m\equiv\rho_{m}/(3\Mp^2H^2)=\Omega_{m0}e^{-3x}/E^2$, $\Omega_k=\Ok e^{-2x}/E^2$, and Eq.~\eqref{eq:fried1} reads
\begin{equation}
\Omega_m+\Omega_d+\Omega_k=1 .
\label{eq:closure}
\end{equation}
Radiation contributes at most about $0.1\%$ of $H^2$ at $z\leq2.5$ and is neglected in the fits; it is added only in the early-universe argument of Sect.~\ref{sec:bbn}.

\textit{The cutoff in a curved universe.} The GO cutoff is~\cite{Granda:2008dk,Granda:2008tm} $L^{-2}=\alpha_{\rm GO}H^2+\gamma_{\rm GO}\dot H$ (we rename GO's $\beta$ to $\gamma_{\rm GO}$ to avoid a clash with $\beta$ in the Tsallis literature). It is built from $H$ and $\dot H$ only. In a curved universe the Ricci scalar is $R=6(\dot H+2H^2+k/a^2)$~\cite{Gao:2007ep,Cai:2008nk}, so the cutoff is not the Ricci cutoff when $\alpha_{\rm GO}=2\gamma_{\rm GO}$, as it is for $k=0$. We therefore consider
\begin{equation}
L^{-2}=\alpha_{\rm GO}H^2+\gamma_{\rm GO}\dot H+\kappa_{\rm GO}\,\frac k{a^2},
\label{eq:gok}
\end{equation}
with two cases: $\kappa_{\rm GO}=0$, the GO cutoff used in curved space by Karami and Fehri~\cite{Karami:2009gof} (our main model), and $\kappa_{\rm GO}=\gamma_{\rm GO}$ with $\alpha_{\rm GO}=2\gamma_{\rm GO}$, the full Ricci-scalar cutoff. The Hubble cutoff ($\gamma_{\rm GO}=\kappa_{\rm GO}=0$) cannot accelerate the Universe~\cite{Hsu:2004ri}.

\subsection{Dimensionless form and the absorption of $c$}
\label{sec:dimless}
Because ${\rm d}/{\rm d}t=H\,{\rm d}/{\rm d}x$,
\begin{equation}
\dot H=H_0^2\,E\,E',\qquad \frac{k}{a^2}=-H_0^2\,\Ok\,e^{-2x}.
\label{eq:hdot}
\end{equation}
Dividing Eq.~\eqref{eq:fried1} by $3\Mp^2H_0^2$, and inserting Eqs.~\eqref{eq:rho}, \eqref{eq:gok} and \eqref{eq:hdot},
\begin{equation}
E^2=\Omega_{m0}e^{-3x}+\Ok e^{-2x}+X\,g\!\left(\frac bX\right),\qquad X\equiv\alpha E^2+\gamma EE'-\kappa\,\Ok e^{-2x},
\label{eq:formA}
\end{equation}
with
\begin{equation}
\alpha\equiv c^2\alpha_{\rm GO},\qquad \gamma\equiv c^2\gamma_{\rm GO},\qquad \kappa\equiv c^2\kappa_{\rm GO},\qquad b\equiv\frac{8\pi^2c^2\delta\Mp^2}{H_0^2}.
\label{eq:absorb}
\end{equation}
To see this, note that $\rho_d/(3\Mp^2H_0^2)=c^2L^{-2}g(u)/H_0^2=Xg(u)$ by Eqs.~\eqref{eq:rho} and~\eqref{eq:gok}, and that $u=8\pi^2\delta\Mp^2L^2=b/X$. The constants $(c,\alpha_{\rm GO},\gamma_{\rm GO},\kappa_{\rm GO},\delta)$ enter only through $(\alpha,\gamma,\kappa,b)$, so $c$ is unobservable and we set $c=1$. Equation~\eqref{eq:formA} is implicit in $E'$.

\subsection{Autonomous system}
\label{sec:formB}
We derive an equivalent explicit system for $(\Omega_d,u,\Omega_k)$ in six steps.

\textit{Step 1 (cutoff fraction).} By Eqs.~\eqref{eq:rho}, \eqref{eq:defs} and~\eqref{eq:formA}, $\Omega_d=Xg(u)/E^2$. Defining $\ell\equiv X/E^2=(LH)^{-2}$,
\begin{equation}
\ell=\frac{\Omega_d}{g(u)}.
\label{eq:ell}
\end{equation}
\textit{Step 2 (deceleration parameter).} Dividing the definition of $X$ by $E^2$ and using $E'/E=\dot H/H^2=-(1+q)$ and $\Ok e^{-2x}/E^2=\Omega_k$ gives $\ell=\alpha-\gamma(1+q)-\kappa\,\Omega_k$. With Eq.~\eqref{eq:ell},
\begin{equation}
q=-1+\frac{\alpha-\kappa\,\Omega_k-\Omega_d/g(u)}{\gamma}.
\label{eq:q}
\end{equation}
\textit{Step 3 (equation of state).} Eq.~\eqref{eq:fried2} divided by $2\Mp^2H^2$ and written with Eq.~\eqref{eq:defs} gives $\dot H/H^2=-\tfrac32\left[\Omega_m+\Omega_d(1+w_d)\right]-\Omega_k$, hence, using Eq.~\eqref{eq:closure},
\begin{equation}
q=\frac12\left(1-\Omega_k+3\Omega_dw_d\right),\qquad w_d=\frac{2q-1+\Omega_k}{3\Omega_d}.
\label{eq:wd}
\end{equation}
\textit{Step 4 (density fractions).} From Eq.~\eqref{eq:cons}, $\rho_d'=-3(1+w_d)\rho_d$, and $H'/H=-(1+q)$, so $\Omega_d'=\Omega_d\left[2q-1-3w_d\right]$. Eliminating $w_d$ with Eq.~\eqref{eq:wd}, and using $\Omega_k'=\Omega_k(-2-2H'/H)$,
\begin{equation}
\Omega_d'=(1-\Omega_d)(1-2q)-\Omega_k,\qquad \Omega_k'=2q\,\Omega_k .
\label{eq:odot}
\end{equation}
The matter fraction obeys $\Omega_m'=\Omega_m(2q-1)$, which is consistent with Eqs.~\eqref{eq:closure} and~\eqref{eq:odot}.

\textit{Step 5 (the deformation).} Writing $\rho_d=3\Mp^2H_0^2\,Xg(b/X)$ and differentiating with respect to $X$ at fixed $b$,
\begin{equation}
\frac{{\rm d}\rho_d}{{\rm d}X}=3\Mp^2H_0^2\left[g(u)-u\,g'(u)\right]=3\Mp^2H_0^2\,g(u)\left[2-\keff(u)\right],
\label{eq:drhodX}
\end{equation}
where the last form uses Eq.~\eqref{eq:keff}. Combined with $\rho_d'=-3(1+w_d)\rho_d$ this gives $X'/X=-3(1+w_d)/(2-\keff)$, and since $u=b/X$,
\begin{equation}
u'=\frac{3u\,(1+w_d)}{2-\keff(u)} .
\label{eq:udot}
\end{equation}
(This is the form of the flat paper~\cite{Pandey:2026flat}; curvature enters only through $w_d$ and $q$. An independent derivation: $u=b/(E^2\ell)$ gives $u'/u=2(1+q)-\ell'/\ell$, with $\ell'/\ell=\Omega_d'/\Omega_d-(\keff-1)u'/u$ and $\Omega_d'/\Omega_d=2q-1-3w_d$; solving for $u'/u$ gives Eq.~\eqref{eq:udot}.)

\textit{Step 6 (initial conditions and expansion rate).} At $x=0$, $E=1$, $\Omega_k=\Ok$, $\Omega_d=1-\Omega_{m0}-\Ok$ and $u=u_0$; Eq.~\eqref{eq:formA} at $x=0$ gives $X_0g(u_0)=\Omega_d(0)$, so
\begin{equation}
\Omega_d(0)=1-\Omega_{m0}-\Ok,\quad \Omega_k(0)=\Ok,\quad u(0)=u_0,\quad b=\frac{\Omega_d(0)\,u_0}{g(u_0)} .
\label{eq:ic}
\end{equation}
Eq.~\eqref{eq:closure} then gives the expansion rate,
\begin{equation}
E^2(x)=\frac{\Omega_{m0}\,e^{-3x}+\Ok\,e^{-2x}}{1-\Omega_d(x)} .
\label{eq:Erec}
\end{equation}
Equations~\eqref{eq:q}, \eqref{eq:wd}, \eqref{eq:odot} and~\eqref{eq:udot} form a closed autonomous system. For $\Ok=0$ it reduces to the flat system~\cite{Pandey:2026flat}; for $b=0$ ($u\equiv0$) it describes the GO-HDE model of Karami and Fehri~\cite{Karami:2009gof}, whose closed-form solution is given in Proposition~\ref{prop:curved}.

\textit{Branch and fold.} By Eq.~\eqref{eq:drhodX}, $\rho_d$ increases with $X$ only while $\keff(u)<2$. Since $\keff$ increases monotonically, this holds for $u<u_*$, with
\begin{equation}
\keff(u_*)=2\iff e^{u_*}(u_*-2)+2=0\iff u_*=1.5936 .
\label{eq:ustar}
\end{equation}
This is the branch connected to HDE ($u\to0$). The density $\rho_d=3\Mp^2H_0^2\,b\,\varphi(u)$, $\varphi=g/u$, has the floor $\rho_d\geq3\Mp^2H_0^2\,b\,\varphi(u_*)$ with $\varphi(u_*)=1.5441$.

\subsection{Tracking solutions}
\label{sec:fixed}
\textit{Matter era.} During matter domination $q=\tfrac12$ and $u\to0$, so $g\to1$ and Eq.~\eqref{eq:q} gives
\begin{equation}
\Omega_d\to\Odm\equiv\alpha-\tfrac32\gamma ,
\label{eq:Odm}
\end{equation}
which is stationary under Eq.~\eqref{eq:odot} because $1-2q=0$ and $\Omega_k\to0$. With radiation added, $q=1$ and, in the same way,
\begin{equation}
\Omega_d\to\Odr\equiv\alpha-2\gamma .
\label{eq:Odr}
\end{equation}

\textit{Curvature era.} If curvature dominates the expansion, $q=0$. Eq.~\eqref{eq:q} then gives, for $\kappa=0$ and $g\to1$,
\begin{equation}
\Omega_d\to\Odk\equiv\alpha-\gamma,\qquad \Omega_k\to1-\Odk,\qquad w_d\to-\tfrac13 ,
\label{eq:Odk}
\end{equation}
where $w_d$ follows from Eq.~\eqref{eq:wd}: $(2q-1+\Omega_k)/(3\Omega_d)=(-\alpha+\gamma)/(3(\alpha-\gamma))$. The dark energy tracks curvature as it tracks matter. Equation~\eqref{eq:udot} then gives $u'=2u$, i.e.
\begin{equation}
u\propto a^3\ \text{(matter)},\qquad u\propto a^4\ \text{(radiation)},\qquad u\propto a^2\ \text{(curvature)},
\label{eq:uscale}
\end{equation}
so the Tsallis deformation affects only the late Universe. The curvature era is reached only if $|\Ok|$ is large enough for curvature to dominate over both matter and dark energy at some epoch; for the small $|\Ok|$ allowed by the data it is a mathematical property of the system (Fig.~\ref{fig:tracker}) rather than a physical epoch, but the coefficient it fixes, $1/(1-\alpha+\gamma)$, controls the observable $H(z)$ at all redshifts (Proposition~\ref{prop:curved}).

\begin{figure}[H]
\centering\pipefig{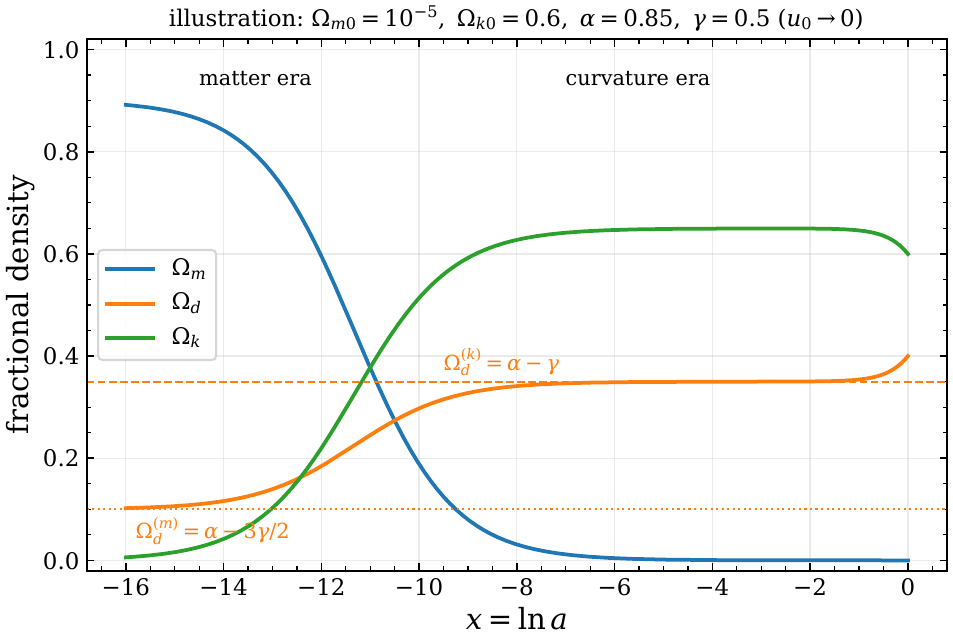}{0.78\textwidth}
\caption{Matter, dark-energy and curvature fractions against $\ln a$ for an illustrative point with a long curvature-dominated era, $\Omega_{m0}=10^{-5}$, $\Ok=0.6$, $\alpha=0.85$, $\gamma=0.5$ (i.e.\ $\wgo=-1.2$, $\Odm=0.1$, $\Odk=0.35$; $u_0\to0$), showing the matter tracker $\Odm=\alpha-3\gamma/2$ (dotted) and the curvature tracker $\Odk=\alpha-\gamma$ (dashed), Eqs.~\eqref{eq:Odm} and~\eqref{eq:Odk}.}
\label{fig:tracker}
\end{figure}

\subsection{Exact results}
\label{sec:exact}

\begin{proposition}[Curved GO-HDE in closed form]
\label{prop:curved}
For $b=0$ and $1-\alpha+\gamma\neq0$, the solution of Eq.~\eqref{eq:formA} with $E(0)=1$ is
\begin{equation}
\begin{gathered}
E^2=\Omt\,a^{-3}+\Okt\,a^{-2}+(1-\Omt-\Okt)\,a^{k},\\
\Omt=\frac{\Omega_{m0}}{1-\Odm},\qquad \Okt=\Rk\,\Ok,\qquad \Rk=\frac{1-\kappa}{1-\alpha+\gamma},
\end{gathered}
\label{eq:Etilde}
\end{equation}
with $k=2(1-\alpha)/\gamma=-3(1+\wgo)$, i.e.\ $\wgo=-1-2(1-\alpha)/(3\gamma)$.
\end{proposition}
\begin{proof}
With $b=0$, $g=1$ and $y\equiv E^2$, so that $EE'=y'/2$, Eq.~\eqref{eq:formA} becomes the linear equation
\begin{equation}
\frac\gamma2\,y'+(\alpha-1)\,y=-\Omega_{m0}e^{-3x}-(1-\kappa)\,\Ok\,e^{-2x}.
\label{eq:linear}
\end{equation}
Trying $y_p=Ae^{-3x}+Be^{-2x}$: the $e^{-3x}$ terms give $A(\alpha-1-\tfrac32\gamma)=-\Omega_{m0}$, so $A=\Omega_{m0}/(1-\Odm)=\Omt$ with Eq.~\eqref{eq:Odm}; the $e^{-2x}$ terms give $B(\alpha-1-\gamma)=-(1-\kappa)\Ok$, so $B=(1-\kappa)\Ok/(1-\alpha+\gamma)=\Rk\Ok$. The homogeneous solution is $Ce^{kx}$ with $k=2(1-\alpha)/\gamma$, and $y(0)=1$ fixes $C=1-\Omt-\Okt$. Since $e^x=a$ this is Eq.~\eqref{eq:Etilde}.
\end{proof}
For $\kappa=0$, Eq.~\eqref{eq:Etilde} is Eq.~(6) of Karami and Fehri~\cite{Karami:2009gof} with their $\beta=\gamma$ and without radiation. In the parameters used below,
\begin{equation}
\Rk\Big|_{\kappa=0}=\frac{3|\wgo|}{(1-\Odm)\,(3|\wgo|-1)},
\label{eq:Rkphys}
\end{equation}
which follows from $1-\alpha+\gamma=(1-\Odm)(3|\wgo|-1)/(3|\wgo|)$ [Eq.~\eqref{eq:map} below]. Subtracting the matter and curvature terms of Eq.~\eqref{eq:formA} from Eq.~\eqref{eq:Etilde}, and using $\Omt-\Omega_{m0}=\Omt\,\Odm$ (from the definition of $\Omt$), gives the dark-energy density of curved GO-HDE:
\begin{equation}
\frac{\rho_d}{3\Mp^2H_0^2}=\Omt\,\Odm\,a^{-3}+(\Rk-1)\,\Ok\,a^{-2}+(1-\Omt-\Okt)\,a^{k}.
\label{eq:rhodecomp}
\end{equation}
The GO dark energy therefore contains a matter-like part, which is the origin of the matter tracker, and a curvature-like part of weight $\Rk-1$, which is the origin of the curvature tracker: if curvature dominates, $\Omega_d\to(\Rk-1)\Ok/(\Rk\Ok)=1-1/\Rk=\alpha-\gamma$ for $\kappa=0$, which is Eq.~\eqref{eq:Odk}. The curvature-like part has the sign of $\Ok$. In a closed universe it is negative, and the holographic density stays positive ($L^2>0$) only where the matter-like and GO-mode parts outweigh it; parameter points with $\Omega_d\le0$ anywhere in the data range are therefore unphysical and are rejected (Sect.~\ref{sec:numerics}). For an open universe no such restriction arises.

\begin{proposition}[Curvature consistency ratio]
\label{prop:ratio}
(i) For $\kappa=0$ and $|\wgo|>1/3$, $\Rk>1$ for every $\Odm\in[0,1)$. (ii) For the Ricci cutoff ($\kappa=\gamma$, $\alpha=2\gamma$), $\Rk=1$ exactly.
\end{proposition}
\begin{proof}
(i) By Eq.~\eqref{eq:Rkphys}, $\Rk=\frac{1}{1-\Odm}\cdot\frac{3|\wgo|}{3|\wgo|-1}$, a product of two factors that are each at least 1 for $\Odm\ge0$ and $|\wgo|>1/3$, and the second is strictly larger than 1. (ii) Inserting $\kappa=\gamma$, $\alpha=2\gamma$ in Eq.~\eqref{eq:Etilde}: $\Rk=(1-\gamma)/(1-2\gamma+\gamma)=1$.
\end{proof}
The physical meaning of Proposition~\ref{prop:ratio} is the central message of this paper. The distance integrals of Sect.~\ref{sec:data} use the \emph{geometric} curvature $\Ok$, because the spatial geometry is that of Eq.~\eqref{eq:metric}. The expansion rate $H(z)$, however, carries the \emph{dynamical} coefficient $\Okt=\Rk\Ok$ of $(1+z)^2$. In $\Lambda$CDM, $w$CDM and CPL the two coincide. In GO-HDE they differ by the factor $\Rk$, because a cutoff built from $H$ and $\dot H$ alone does not contain the curvature term $k/a^2$ that completes $\dot H+2H^2$ to the Ricci scalar: the dark energy responds to $k/a^2$ only indirectly, through $\dot H$, and this response adds the curvature-like part $(\Rk-1)\Ok a^{-2}$ of Eq.~\eqref{eq:rhodecomp} to $H^2$. Completing the cutoff with the Ricci term restores $\Rk=1$. On the nucleosynthesis surface of Proposition~\ref{prop:bbn} below,
\begin{equation}
\Rk\Big|_{\kappa=0,\ \Odm=1/(1+3|\wgo|)}=\frac{1+3|\wgo|}{3|\wgo|-1},
\label{eq:Rkbbn}
\end{equation}
which equals $2$ for $\wgo=-1$ [from Eq.~\eqref{eq:Rkphys} with $1-\Odm=3|\wgo|/(1+3|\wgo|)$]. Figure~\ref{fig:ratio} maps $\Rk$.

\begin{figure}[H]
\centering\pipefig{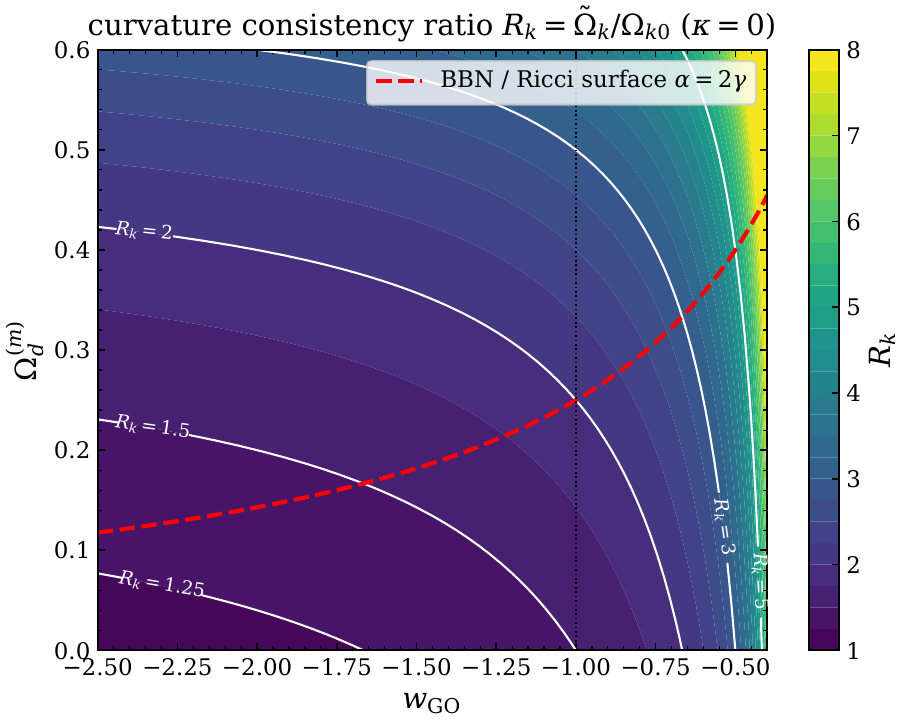}{0.72\textwidth}
\caption{Curvature consistency ratio $\Rk=\Okt/\Ok$, Eq.~\eqref{eq:Rkphys}, in the $(\wgo,\Odm)$ plane for $\kappa=0$ (colour, clipped at 8; white contours $\Rk=1.25,1.5,2,3,5$). $\Rk>1$ everywhere (Proposition~\ref{prop:ratio}); its smallest value in the box, $1.15$, is at $\wgo=-2.5$, $\Odm=0$, and $\Rk\to\infty$ as $|\wgo|\to1/3$. Dashed red line: the nucleosynthesis (Ricci) surface $\Odm=1/(1+3|\wgo|)$, Eq.~\eqref{eq:lock}, on which $\Rk=(1+3|\wgo|)/(3|\wgo|-1)$; dotted line: $\wgo=-1$.}
\label{fig:ratio}
\end{figure}

\textit{The Clarkson test cannot see $\Rk$.} For any FLRW metric with $D\equiv H_0D_M/c_0$ and $D'\equiv{\rm d}D/{\rm d}z$,
\begin{equation}
\frac{E^2D'^2-1}{D^2}=\Ok,
\label{eq:clarkson}
\end{equation}
independently of the dark energy, because $D=\sinh(\sqrt{\Ok}\,\chi)/\sqrt{\Ok}$ with $\chi=\int{\rm d}z/E$ gives $D'=\cosh(\sqrt{\Ok}\chi)/E$ and $\cosh^2-1=\sinh^2$ (for $\Ok<0$ replace $\sinh\to\sin$, $\cosh\to\cos$). The curvature test of Clarkson, Bassett and Lu~\cite{Clarkson:2007pz} therefore returns the geometric $\Ok$ in every model, including this one. What is specific to GO holography is that the coefficient of $(1+z)^2$ in the \emph{same} $E(z)$ is $\Rk\Ok$; a fit that assumes $\Okt=\Ok$ (as in $w$CDM with curvature) recovers neither value correctly. Figure~\ref{fig:clarkson} shows both.

\begin{figure}[H]
\centering\pipefig{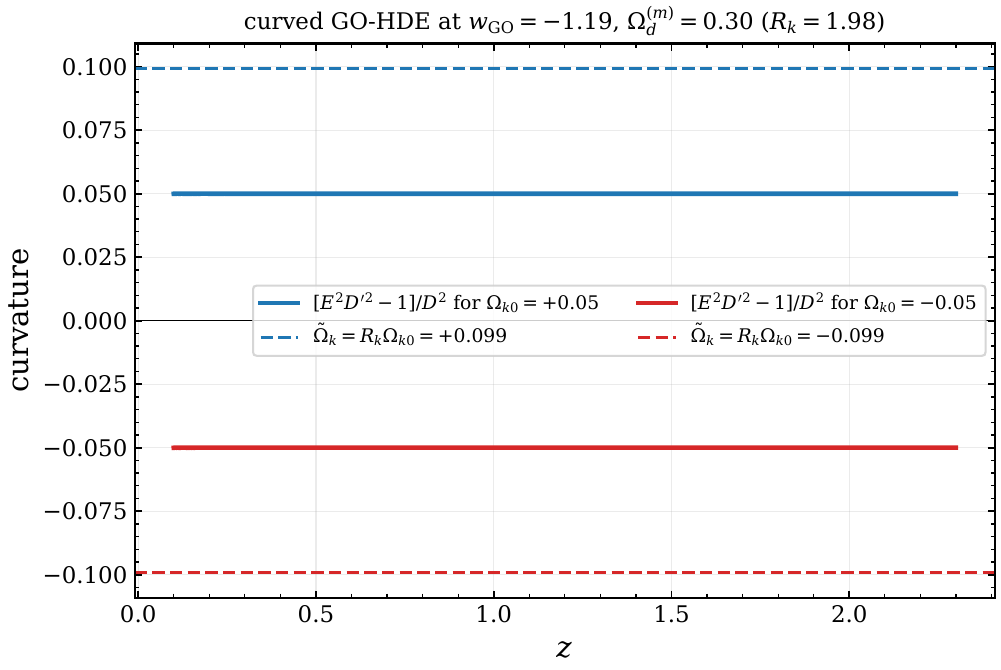}{0.72\textwidth}
\caption{Clarkson test function $[E^2D'^2-1]/D^2$, computed numerically from $E(z)$ and $D(z)$ of curved GO-HDE with $\wgo$ and $\Odm$ at the joint best fit, for an open ($\Ok=+0.05$, blue) and a closed ($\Ok=-0.05$, red) universe (solid lines: constant and equal to $\Ok$, Eq.~\eqref{eq:clarkson}). Dashed: the coefficient $\Okt=\Rk\Ok$ of $(1+z)^2$ in $E^2(z)$ of the same models.}
\label{fig:clarkson}
\end{figure}

\begin{proposition}[First order in the deformation]
\label{prop:first}
To first order in $b$,
\begin{equation}
E^2=E^2_{\rm GO\text{-}HDE}+\frac{b}{2(1-\alpha)}\left(1-a^{k}\right)+O(b^2),
\label{eq:firstorder}
\end{equation}
independently of the curvature.
\end{proposition}
\begin{proof}
For small $u=b/X$, $g\simeq1+u/2$ and $Xg(b/X)\simeq X+b/2$. The constant $b/2$ on the right-hand side of Eq.~\eqref{eq:linear} adds the particular solution $D=b/(2(1-\alpha))$ [from $(\alpha-1)D=-b/2$]; with $y(0)=1$ the constant $C$ changes by $-D$, which gives Eq.~\eqref{eq:firstorder}.
\end{proof}
At first order the Tsallis correction acts as a cosmological constant whose weight relative to matter decays as $(1+z)^{-3}$; it does not know about curvature. This explains why $u_0$ and curvature are separately weakly constrained.

\textit{Parameter map.} Equations~\eqref{eq:Odm} and~\eqref{eq:Etilde} show that late-time data measure $(\Omt,\wgo,\Ok)$ and the combination $\Rk\Ok$, while $\Odm$ appears only through $\Omt$ and $\Rk$ [Eq.~\eqref{eq:Rkphys}]. Unlike the flat case, $\Odm$ is therefore not a perfectly blind direction when $\Ok\neq0$: it changes $\Okt$ at fixed $\Ok$. We fit the physical set $(\Omt,\Ok,\wgo,\Odm,u_0)$ and recover $\alpha,\gamma$ by inverting Eqs.~\eqref{eq:Odm} and~\eqref{eq:Etilde}:
\begin{equation}
\gamma=\frac{2(1-\Odm)}{3|\wgo|},\qquad \alpha=1+\tfrac32\gamma(1+\wgo),\qquad \Omega_{m0}=\Omt(1-\Odm).
\label{eq:map}
\end{equation}
(Check: Eq.~\eqref{eq:Odm} gives $\Odm=\alpha-\tfrac32\gamma$; inserting $\alpha$ gives $\Odm=1+\tfrac32\gamma\wgo$, i.e.\ the first relation for $\wgo<0$.) The prior $\wgo<-0.4$ keeps $1-\alpha+\gamma=(1-\Odm)(3|\wgo|-1)/(3|\wgo|)>0$ and avoids the pole of $\Rk$ at $|\wgo|=1/3$.

\label{sec:bbn}
\begin{proposition}[Nucleosynthesis and the Ricci lock persist]
\label{prop:bbn}
During radiation domination the GO tracker counts as extra relativistic species,
\begin{equation}
\Delta N_\nu=\kappa_{\rm BBN}\,\frac{\Odr}{1-\Odr},\qquad \kappa_{\rm BBN}=\frac{43}{7}\ \ (T\simeq1~{\rm MeV}),
\label{eq:dnu}
\end{equation}
independently of curvature. BBN therefore requires $\Odr\simeq0$, i.e.\ $\alpha=2\gamma$, and on that surface
\begin{equation}
\Odm=\frac{1}{1+3|\wgo|}.
\label{eq:lock}
\end{equation}
\end{proposition}
\begin{proof}
Curvature falls as $a^{-2}$ while radiation falls as $a^{-4}$, so $\Omega_k\to0$ in the radiation era and the argument of the flat case applies. Adding radiation $\rho_r\propto a^{-4}$ to Eq.~\eqref{eq:linear} adds a particular solution $\propto e^{-4x}$ (the solution with matter, radiation and curvature is Eq.~(6) of Ref.~\cite{Karami:2009gof}); during radiation domination $\rho_d/(\rho_d+\rho_r)=\Odr$, i.e.\ $\rho_d/\rho_r=\Odr/(1-\Odr)$, and $u\propto a^4$ [Eq.~\eqref{eq:uscale}] is negligible. At $T\simeq1$~MeV, $\rho_r=(\pi^2/30)g_*T^4$ with $g_*=10.75$, and one neutrino species contributes $\tfrac78\times2=\tfrac74$~\cite{Kolb:1990vq,baumann2022cosmology}; an extra component counts as $\Delta N_\nu=(\rho_d/\rho_r)g_*/(7/4)$, which is Eq.~\eqref{eq:dnu}. Setting $\Odr=0$ gives $\alpha=2\gamma$; inserting in Eq.~\eqref{eq:map}, $2\gamma=1+\tfrac32\gamma(1+\wgo)$, so $\gamma=2/(1+3|\wgo|)$ and $\Odm=\alpha-\tfrac32\gamma=\gamma/2$.
\end{proof}
The BBN measurement $N_\nu=2.86\pm0.15$~\cite{Fields:2019pfx} (the erratum~\cite{Fields:2020err} gives $2.843\pm0.154$, a negligible change here; consistent with the BBN+CMB value~\cite{Yeh:2022heq}) allows $-0.08\lesssim\Odr\lesssim0.03$ at $2\sigma$. On the BBN surface with $\kappa=0$ the curvature ratio is Eq.~\eqref{eq:Rkbbn}: BBN thus makes $\Rk$ a function of $\wgo$ alone: $\Rk=2$ at $\wgo=-1$, $\Rk>2$ for $-1<\wgo<-1/3$ (a GO mode that does not accelerate strongly), and $\Rk\to1$ as $|\wgo|\to\infty$, so a phantom GO mode ($\wgo<-1$) has $1<\Rk<2$.

\begin{proposition}[The fold singularity survives curvature]
\label{prop:fold}
If a trajectory with $1+w_d>0$ reaches $u=u_*$ at $x=x_f$, then $H$ and $\dot H$ stay finite while $\ddot H\propto(t_f-t)^{-1/2}$.
\end{proposition}
\begin{proof}
Near the fold $2-\keff\simeq\keff'(u_*)(u_*-u)$ with $\keff'(u_*)>0$, so Eq.~\eqref{eq:udot} reads $u'\simeq\mathcal A/(u_*-u)$, $\mathcal A=3u_*(1+w_d)/\keff'(u_*)>0$, whence $(u_*-u)^2=2\mathcal A(x_f-x)$ and
\begin{equation}
u_*-u\propto(x_f-x)^{1/2}.
\label{eq:foldscaling}
\end{equation}
By Eqs.~\eqref{eq:q}, \eqref{eq:wd} and~\eqref{eq:odot}, $\Omega_d$, $\Omega_k$ and $q$ remain finite, and so do $H$ and $\dot H=-H^2(1+q)$. Differentiating Eq.~\eqref{eq:q}, $q'=-\gamma^{-1}\left[\ell'+\kappa\Omega_k'\right]$ with $\ell'=\ell[\Omega_d'/\Omega_d-(\keff-1)u'/u]\propto u'\propto(x_f-x)^{-1/2}$, since $\Omega_d'$ and $\Omega_k'$ are finite. Then ${\rm d}/{\rm d}t=H\,{\rm d}/{\rm d}x$ and $H'/H=-(1+q)$ give
\begin{equation}
\ddot H=H^3\left[2(1+q)^2-q'\right]\propto(t_f-t)^{-1/2},
\label{eq:hddot}
\end{equation}
using $x_f-x\simeq H_f(t_f-t)$.
\end{proof}
This is a generalized sudden singularity in the sense of Barrow~\cite{Barrow:2004xh,Barrow:2004hk}: energy density and pressure stay finite while a higher derivative of $H$ diverges~\cite{Nojiri:2005sx}. Dark energy that keeps diluting ($w_d>-1$) reaches the floor of the density, which it cannot cross.

\subsection{Background evolution}
\label{sec:evolution}
Figure~\ref{fig:params} shows $\Omega_d(z)$, $w_d(z)$ and $q(z)$ for $-1<z<3$, obtained by integrating Eqs.~\eqref{eq:q}--\eqref{eq:udot} from Eq.~\eqref{eq:ic} into the past and the future, varying one parameter at a time about the joint best fit of Sect.~\ref{sec:results}. Every curve has a matter era with $\Omega_d\to\Odm$ and $w_d\to0$, a transition from deceleration to acceleration, and dark-energy domination today. Figure~\ref{fig:params} shows four effects. Increasing $\Ok$ lowers $\Omega_d$ today, because $\Omega_d(0)=1-\Omega_{m0}-\Ok$ [Eq.~\eqref{eq:ic}], raises $w_d$ and $q$ today, and brings the fold earlier; in the past the curves converge, because curvature falls behind matter. A more negative $\wgo$ pushes $w_d$ towards and below $-1$ and delays the fold. $\Odm$ sets the dark-energy fraction in the past, as Eq.~\eqref{eq:Odm} requires, and has little effect on $q(z)$ at $z\lesssim1$. Increasing $u_0$ raises $w_d$ today above the value of GO-HDE ($u_0=0$), for which $w_d$ crosses $-1$ in the future, and moves the fold closer to the present.

\begin{figure}[H]
\centering\pipefig{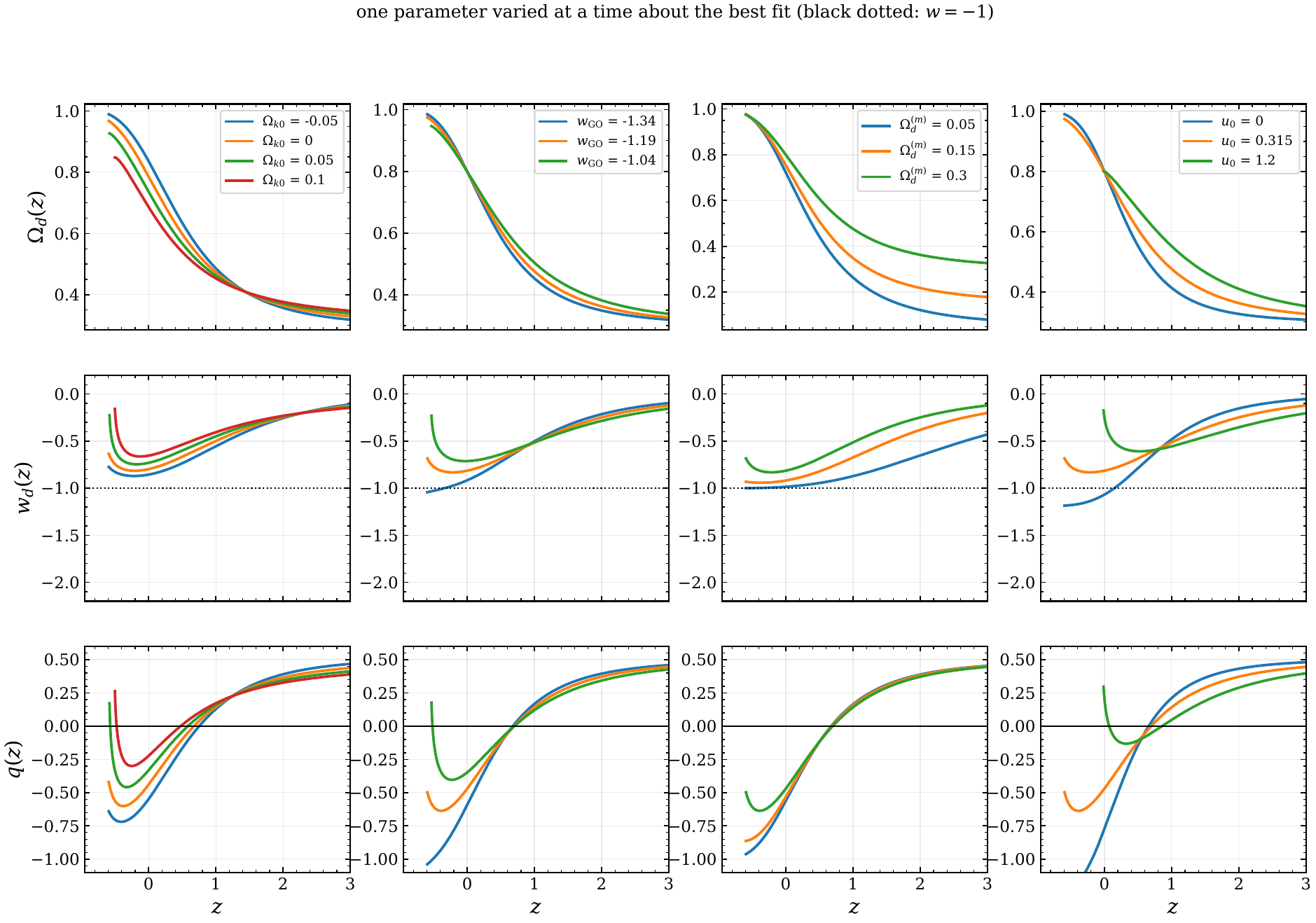}{\textwidth}
\caption{Evolution of $\Omega_d$ (top), $w_d$ (middle) and $q$ (bottom) for $-1<z<3$; $z<0$ is the future. Each column varies one parameter about the joint best fit of curved GO-NTHDE: $\Ok$, $\wgo$, $\Odm$, $u_0$ (legends). Dotted: $w=-1$.}
\label{fig:params}
\end{figure}

\section{Observational data and likelihoods}
\label{sec:data}

\subsection{Cosmological distances in a curved universe}
\label{sec:dist}
With $\chi(z)\equiv\int_0^z{\rm d}z'/E(z')$ the dimensionless line-of-sight comoving distance, the transverse comoving distance in the geometry~\eqref{eq:metric} is~\cite{Hogg:1999ad}
\begin{equation}
D_M(z)=\frac{c_0}{H_0}\times\begin{cases}\dfrac{\sinh\!\big(\sqrt{\Ok}\,\chi\big)}{\sqrt{\Ok}}&\Ok>0\ \text{(open)},\\[2mm] \chi&\Ok=0,\\[1mm] \dfrac{\sin\!\big(\sqrt{-\Ok}\,\chi\big)}{\sqrt{-\Ok}}&\Ok<0\ \text{(closed)},\end{cases}
\label{eq:dm}
\end{equation}
with $c_0=299\,792.458$~km\,s$^{-1}$ the speed of light. The luminosity and Hubble distances are
\begin{equation}
D_L=(1+z)D_M,\qquad D_H=\frac{c_0}{H(z)},
\label{eq:dists}
\end{equation}
and BAO analyses also use the volume-averaged distance~\cite{SDSS:2005xqv}
\begin{equation}
D_V(z)=\left[z\,D_M^2(z)\,D_H(z)\right]^{1/3}.
\label{eq:dv}
\end{equation}
The distance modulus is $\mu(z)=5\log_{10}[D_L(z)/{\rm Mpc}]+25$~\cite{Hogg:1999ad}. Here $\Ok$ is the \emph{geometric} curvature of Eq.~\eqref{eq:metric}, while $E(z)$ is the expansion rate of the model, which contains the \emph{dynamical} coefficient $\Okt=\Rk\Ok$ of Proposition~\ref{prop:curved}. For $w$CDM, CPL and $\Lambda$CDM with curvature the two coincide. The integral $\chi$ is evaluated with the trapezoidal rule on 3000 points.

\subsection{Type Ia supernovae}
\label{sec:sn}
The Pantheon+SH0ES sample~\cite{Scolnic:2021amr,Brout:2022vxf} gives, for each supernova $i$, a standardized distance modulus $\mu_i^{\rm obs}$ computed with the SH0ES absolute magnitude. We use the 1657 light curves with $z_{\rm HD}>0.01$ or a Cepheid host (77 calibrators), with the full statistical and systematic covariance $C_{\rm SN}$. The residuals are
\begin{equation}
\Delta_i=\begin{cases}\mu_i^{\rm obs}-\mu_i^{\rm Ceph} & \text{(Cepheid hosts)},\\ \mu_i^{\rm obs}-\mu(z_i) & \text{(all others)},\end{cases}
\label{eq:snres}
\end{equation}
where $\mu_i^{\rm Ceph}$ are the Cepheid distances of SH0ES~\cite{Riess:2021jrx} and $\mu(z)$ is evaluated at the Hubble-diagram redshift $z_{\rm HD}$ in both factors of Eq.~\eqref{eq:dists}. A shift of the absolute magnitude moves all $\mu_i^{\rm obs}$ by the same amount. Marginalizing the likelihood over this shift with a flat prior gives~\cite{Goliath:2001af,Nesseris:2005ur}
\begin{equation}
\chi^2_{\rm SN}=\Delta^{\!\top}C_{\rm SN}^{-1}\Delta-\frac{\left(\sum_{ij}(C_{\rm SN}^{-1})_{ij}\Delta_j\right)^2}{\sum_{ij}(C_{\rm SN}^{-1})_{ij}} .
\label{eq:chisn}
\end{equation}
To derive Eq.~\eqref{eq:chisn}, let all residuals shift by the same unknown amount $m$, $\Delta\to\Delta-m\mathbf 1$. Then
\begin{equation}
\chi^2(m)=\mathcal A-2m\mathcal B+m^2\mathcal C=\mathcal A-\frac{\mathcal B^2}{\mathcal C}+\mathcal C\left(m-\frac{\mathcal B}{\mathcal C}\right)^2,
\label{eq:margderiv}
\end{equation}
with $\mathcal A=\Delta^{\!\top}C_{\rm SN}^{-1}\Delta$, $\mathcal B=\mathbf 1^{\!\top}C_{\rm SN}^{-1}\Delta$ and $\mathcal C=\mathbf 1^{\!\top}C_{\rm SN}^{-1}\mathbf 1$. Integrating $e^{-\chi^2(m)/2}$ over $m$ gives $\sqrt{2\pi/\mathcal C}\,e^{-(\mathcal A-\mathcal B^2/\mathcal C)/2}$; the prefactor does not depend on the cosmological parameters, so the marginalized $\chi^2$ is $\mathcal A-\mathcal B^2/\mathcal C$. Minimizing $\chi^2(m)$ over $m$ gives the same expression.
Because the Cepheid distances do not depend on $H_0$, the calibrators keep the absolute distance scale and the supernovae constrain $H_0$. As alternatives we use Union3 (22 binned distance moduli with covariance)~\cite{Rubin:2023ovl}, DES-Dovekie (1820 supernovae)~\cite{Popovic:2025dovekie}, and PantheonPlus without SH0ES, i.e.\ the 1590 supernovae at $z>0.01$ with $\Delta_i=\mu_i^{\rm obs}-\mu(z_i)$ for all of them; Eq.~\eqref{eq:chisn} is applied in every case, so without calibrators $H_0$ is not constrained by supernovae.

\subsection{Cosmic chronometers}
Passively evolving galaxies measure the expansion rate from their differential ages~\cite{Jimenez:2001gg},
\begin{equation}
H(z)=-\frac{1}{1+z}\frac{{\rm d}z}{{\rm d}t}.
\label{eq:cc}
\end{equation}
We use 32 such measurements at $0.07\leq z\leq1.965$~\cite{Moresco:2022phi} with
\begin{equation}
\chi^2_{\rm CC}=\sum_{i=1}^{32}\frac{\left[H_i^{\rm obs}-H_0E(z_i)\right]^2}{\sigma_i^2}.
\label{eq:chicc}
\end{equation}
As a robustness check we replace the diagonal errors by the covariance $C_{\rm CC}={\rm diag}(\sigma_i^2)+f^2H^{\rm obs}(H^{\rm obs})^{\!\top}$ with $f=0.05$, a fully correlated $5\%$ systematic.

\subsection{Baryon acoustic oscillations}
DESI DR2~\cite{DESI:2025zgx} measures $D_M/r_d$ and $D_H/r_d$ at effective redshifts $z=0.510$, $0.706$, $0.934$, $1.321$, $1.484$ and $2.330$, and $D_V/r_d$ at $z=0.295$: 13 numbers with covariance $C_{\rm BAO}$. Here $r_d$ is the comoving sound horizon at the baryon-drag epoch $z_d$~\cite{Eisenstein:1997ik},
\begin{equation}
r_d=\int_{z_d}^{\infty}\frac{c_s(z)}{H(z)}\,{\rm d}z,
\label{eq:rd}
\end{equation}
with $c_s$ the sound speed of the photon--baryon fluid. Computing $r_d$ needs early-universe physics; we instead treat $r_d$ as a free parameter, so that BAO constrains ratios of distances and the product $H_0r_d$. With $v$ the vector of the 13 ratios, computed from Eqs.~\eqref{eq:dm}--\eqref{eq:dv},
\begin{equation}
\chi^2_{\rm BAO}=\left(v^{\rm obs}-v^{\rm th}\right)^{\!\top}C_{\rm BAO}^{-1}\left(v^{\rm obs}-v^{\rm th}\right).
\label{eq:chibao}
\end{equation}

\subsection{Nucleosynthesis}
When BBN is used, it enters as one data point,
\begin{equation}
\chi^2_{\rm BBN}=\frac{\left(3+\Delta N_\nu-N_\nu^{\rm obs}\right)^2}{\sigma_\nu^2},\qquad N_\nu^{\rm obs}=2.86,\ \sigma_\nu=0.15~\text{\cite{Fields:2019pfx}},
\label{eq:chibbn}
\end{equation}
with $\Delta N_\nu$ from Eq.~\eqref{eq:dnu}. Models without a radiation tracker have $\Delta N_\nu=0$.

\subsection{External curvature prior}
As a check that is \emph{not} part of the main result, we add one Gaussian term for the curvature measured by Planck 2018 with BAO in $\Lambda$CDM+$\Ok$~\cite{Planck:2018vyg},
\begin{equation}
\chi^2_{\Omega_k}=\frac{(\Ok-\bar\Omega_k)^2}{\sigma_k^2},\qquad \bar\Omega_k=0.0007,\ \sigma_k=0.0019 .
\label{eq:chiok}
\end{equation}
the 68\% result for TT,TE,EE+lowE+lensing+BAO. It is used only in the robustness run labelled ``Planck+BAO Gaussian prior on $\Ok$'' and enters the flat models with $\Ok=0$ exactly. Because the term is added to a likelihood whose prior box on $\Ok$ is $0.6$ wide, a curved model pays an Occam factor of about $\ln[0.6/(\sqrt{2\pi}\,\sigma_k)]\simeq4.8$ in $\ln Z$ relative to its flat counterpart; this is the price of a parameter that the external data fix, and it is the reason to read this run together with the main one rather than instead of it.

\subsection{Joint likelihood}
The data sets are independent, so
\begin{equation}
\chi^2_{\rm tot}=\chi^2_{\rm SN}+\chi^2_{\rm CC}+\chi^2_{\rm BAO}\,(+\chi^2_{\rm BBN})\,(+\chi^2_{\Omega_k}),\qquad \mathcal L(D|\theta)\propto e^{-\chi^2_{\rm tot}/2},
\label{eq:like}
\end{equation}
with $N=1657+32+13=1702$ data points in the main analysis.


\section{Statistical methodology}
\label{sec:stats}

For each parameter point the background equations of Sect.~\ref{sec:model} are integrated, the observables of Sect.~\ref{sec:data} are computed and compared with the data, and the resulting likelihood is explored with the methods described in this section.

\subsection{Bayesian inference}
For a model $M$ with parameters $\theta$, Bayes' theorem gives the posterior~\cite{Trotta:2008qt}
\begin{equation}
P(\theta|D,M)=\frac{\mathcal L(D|\theta,M)\,\pi(\theta|M)}{Z(D|M)},\qquad Z(D|M)=\int\mathcal L(D|\theta,M)\,\pi(\theta|M)\,{\rm d}^k\theta ,
\label{eq:bayes}
\end{equation}
where $\pi$ is the prior and the evidence $Z$ normalizes the posterior. The posterior gives parameter constraints; the evidence compares models through the Bayes factor
\begin{equation}
\ln B_{12}=\ln Z_1-\ln Z_2 .
\label{eq:bf}
\end{equation}
We use the Jeffreys scale as calibrated in Ref.~\cite{Trotta:2008qt} (see also Ref.~\cite{Kass:1995loi}): $|\ln B|<1$ inconclusive, $1$--$2.5$ weak, $2.5$--$5$ moderate, and $>5$ strong evidence.

\subsection{Parameters and priors}
\label{sec:priors}
The models and their parameter vectors are:
\begin{itemize}
\item flat $\Lambda$CDM, flat $w$CDM, flat GO-HDE and flat GO-NTHDE (the reference models of the flat paper~\cite{Pandey:2026flat}, re-run with the priors of this paper);
\item o$\Lambda$CDM: $(\Omega_m,\allowbreak \Ok,\allowbreak H_0,\allowbreak r_d)$, with $E^2=\Omega_ma^{-3}+\Ok a^{-2}+(1-\Omega_m-\Ok)$;
\item o$w$CDM: $(\Omega_m,\allowbreak \Ok,\allowbreak w,\allowbreak H_0,\allowbreak r_d)$; oCPL: $(\Omega_m,\allowbreak \Ok,\allowbreak w_0,\allowbreak w_a,\allowbreak H_0,\allowbreak r_d)$ with $w(a)=w_0+w_a(1-a)$~\cite{Chevallier:2000qy,Linder:2002et};
\item curved GO-HDE ($u_0=0$, $\kappa=0$): $(\Omt,\allowbreak \Ok,\allowbreak \wgo,\allowbreak \Odm,\allowbreak H_0,\allowbreak r_d)$;
\item curved GO-NTHDE (main model) and its truncation T2: $(\Omt,\allowbreak \Ok,\allowbreak \wgo,\allowbreak \Odm,\allowbreak u_0,\allowbreak H_0,\allowbreak r_d)$;
\item curved Ricci-NTHDE: $(\Omt,\allowbreak \Ok,\allowbreak \wgo,\allowbreak u_0,\allowbreak H_0,\allowbreak r_d)$, with $\kappa=\gamma$, $\alpha=2\gamma$ and $\Odm$ fixed by Eq.~\eqref{eq:lock}.
\end{itemize}
We use the physical parametrization of Eq.~\eqref{eq:map}. The flat priors are listed in Table~\ref{tab:priors}.

\begin{table}[H]
\centering\small
\caption{Flat priors. $u_*=1.5936$. The upper edge of $w$ and $\wgo$ is the same ($-0.4$) in all models, so the evidences of flat GO-HDE and flat $w$CDM must agree (Proposition~\ref{prop:curved} at $\Ok=0$).}
\label{tab:priors}
\begin{tabular}{ll}
\toprule
Parameter & Prior\\
\midrule
$\Omt$ (GO models), $\Omega_m$ & $[0.1,0.5]$\\
$\Ok$ & $[-0.3,0.3]$; $[-0.6,0.6]$ and $[-0.1,0.1]$ as checks\\
$\wgo$, $w$ & $[-2.5,-0.4]$\\
$w_0$, $w_a$ (oCPL) & $[-3,1]$, $[-3,2]$, with $w_0+w_a<0$\\
$\Odm$ & $[0,0.3]$; $[0,0.6]$ with BBN\\
$u_0$ & $[0,u_*)$; halved as a check\\
$H_0$ [km\,s$^{-1}$Mpc$^{-1}$], $r_d$ [Mpc] & $[60,85]$, $[120,180]$\\
\bottomrule
\end{tabular}
\end{table}

Points of the prior box $\Pi$ without a physical solution (for instance with a fold inside $z\leq2.5$) have $\mathcal L=0$, so Eq.~\eqref{eq:bayes} normalizes the evidence over $\Pi$. Normalizing over the allowed region $\mathcal A\subset\Pi$ instead gives
\begin{equation}
Z_{\mathcal A}=\frac{Z_\Pi}{f},\qquad f=\int_{\mathcal A}\pi\,{\rm d}^k\theta\simeq\frac{N_{\rm valid}}{N_{\rm draw}},
\label{eq:volume}
\end{equation}
where $f$ is estimated from $N_{\rm draw}=4000$ random prior draws. We quote both normalizations.

\subsection{Maximum likelihood and Markov chain Monte Carlo}
Best fits minimize $\chi^2_{\rm tot}$ with Powell's derivative-free method~\cite{Powell:1964} from several starting points. Posteriors are sampled with the affine-invariant ensemble sampler of Goodman and Weare~\cite{Goodman:2010} in the \texttt{emcee} implementation~\cite{Foreman-Mackey:2012any}. Each walker $X_k$ is moved along the line through another walker $X_j$,
\begin{equation}
Y=X_j+\mathcal Z\,(X_k-X_j),\qquad p(\mathcal Z)\propto\mathcal Z^{-1/2}\ \ \text{on}\ \ [1/a_s,a_s],\ a_s=2,
\label{eq:stretch}
\end{equation}
and the move is accepted with probability $\min\{1,\mathcal Z^{k-1}P(Y)/P(X_k)\}$, $k$ being the number of parameters~\cite{Goodman:2010}. We run 48 walkers for at least $2\times10^4$ steps and discard the first quarter as burn-in. Convergence requires, for every parameter, at least 50 integrated autocorrelation times,
\begin{equation}
\tau=1+2\sum_{t\geq1}\rho(t),
\label{eq:tau}
\end{equation}
with $\rho(t)$ the normalized autocorrelation function at lag $t$~\cite{Sokal:1997,Goodman:2010}, and a Gelman--Rubin statistic $\hat R<1.01$~\cite{Gelman:1992zz}. For $\hat R$ we split the walkers into two groups treated as chains of length $n$, with mean within-chain variance $W$ and between-chain variance $B$:
\begin{equation}
\hat R=\sqrt{\frac{\tfrac{n-1}{n}W+\tfrac1nB}{W}} .
\label{eq:rhat}
\end{equation}

\subsection{Nested sampling and the evidence}
Nested sampling rewrites the evidence of Eq.~\eqref{eq:bayes} as a one-dimensional integral~\cite{Skilling:2006gxv},
\begin{equation}
Z=\int_0^1\mathcal L(X)\,{\rm d}X,\qquad X(\lambda)=\int_{\mathcal L(\theta)>\lambda}\pi(\theta)\,{\rm d}^k\theta ,
\label{eq:ns}
\end{equation}
where $X$ is the prior mass with likelihood above $\lambda$. Starting from $n_{\rm live}$ prior points, the point of lowest likelihood $\mathcal L_i$ is replaced at each step by a new prior point with higher likelihood; the enclosed prior mass then shrinks on average as $X_i\simeq e^{-i/n_{\rm live}}$, and
\begin{equation}
Z\simeq\sum_i\mathcal L_i\,(X_{i-1}-X_i).
\label{eq:nssum}
\end{equation}
We use \texttt{dynesty}~\cite{Speagle:2019ivv} with $n_{\rm live}=500$, multi-ellipsoid bounds~\cite{Feroz:2008xx} and random-walk proposals, and stop when the estimated remaining evidence satisfies $\Delta\ln Z<0.1$. The statistical error of $\ln Z$ is about $0.2$.

\subsection{Information criteria}
We also quote~\cite{Akaike:1974,Schwarz:1978,Spiegelhalter:2002,Gelman:2004BDA}
\begin{equation}
\begin{gathered}
{\rm AIC}=\chi^2_{\min}+2k,\qquad {\rm BIC}=\chi^2_{\min}+k\ln N,\\
{\rm DIC}=\bar D+p_V,\qquad p_V=\tfrac12{\rm Var}[D],
\end{gathered}
\label{eq:ic_crit}
\end{equation}
where $D(\theta)=\chi^2_{\rm tot}(\theta)$ and $\bar D$, ${\rm Var}[D]$ are its posterior mean and variance. Differences are taken relative to $\Lambda$CDM; negative values favour the model. AIC and BIC charge every parameter in full, whereas $p_V$ counts only parameters the data constrain~\cite{Gelman:2004BDA}, so blind directions such as $\Odm$ are not penalized by DIC.

\subsection{Consistency of the data sets}
Before combining the supernovae with chronometers and BAO, we compute the tension statistic~\cite{Raveri:2018wln}
\begin{equation}
Q_{\rm DMAP}=\chi^2_{\min}({\rm SN+CC+BAO})-\chi^2_{\min}({\rm SN})-\chi^2_{\min}({\rm CC+BAO}),
\label{eq:qdmap}
\end{equation}
which follows a $\chi^2$ distribution with $N_{\rm SN}+N_{\rm CC+BAO}-N_{\rm joint}$ degrees of freedom when the data sets agree. Here
\begin{equation}
N_{\rm eff}=k-{\rm tr}\!\left(C_\Pi^{-1}C_{\rm post}\right)
\label{eq:neff}
\end{equation}
is the number of parameters constrained beyond the prior ($C_\Pi$: prior covariance, $({\rm hi}-{\rm lo})^2/12$ for a flat prior; $C_{\rm post}$: posterior covariance)~\cite{Raveri:2018wln}. The probability $p$ of exceeding $Q_{\rm DMAP}$ is converted into an equivalent Gaussian significance $n_\sigma=\Phi^{-1}(1-p/2)$, with $\Phi$ the standard normal distribution function.

\subsection{Derived quantities}
From each posterior sample we compute $q(z)$ and $w_d(z)$ from Eqs.~\eqref{eq:q} and \eqref{eq:wd}; the curvature consistency ratio $\Rk$ and the dynamical curvature $\Okt=\Rk\Ok$ from Eq.~\eqref{eq:Rkphys}; the tracker fractions $\Odk=\alpha-\gamma$ and $\Odr=\alpha-2\gamma$; the probability of an open universe $P(\Ok>0)$; the transition redshift $z_t$ from $q(z_t)=0$; and the CPL-equivalent parameters~\cite{Chevallier:2000qy,Linder:2002et}
\begin{equation}
w_0=w_d(z=0),\qquad w_a=-\left.\frac{{\rm d}w_d}{{\rm d}a}\right|_{a=1}=-w_d'(x=0),
\label{eq:w0wa}
\end{equation}
where $w_d'$ is evaluated analytically from the state $(\Omega_d,u)$ today. Differentiating Eq.~\eqref{eq:wd}, and Eq.~\eqref{eq:q} with $\ell=\Omega_d/g(u)$,
\begin{equation}
w_d'=\frac{2q'+\Omega_k'}{3\Omega_d}-\frac{w_d\,\Omega_d'}{\Omega_d},\qquad q'=-\frac{\ell'+\kappa\,\Omega_k'}{\gamma},\qquad \ell'=\ell\left[\frac{\Omega_d'}{\Omega_d}-(\keff-1)\frac{u'}{u}\right],
\label{eq:wprime}
\end{equation}
with $\Omega_d'$, $\Omega_k'$ and $u'/u$ from Eqs.~\eqref{eq:odot} and \eqref{eq:udot}; this follows by differentiating the second form of Eq.~\eqref{eq:wd}. From Eq.~\eqref{eq:absorb} with $c=1$, $\delta=bH_0^2/(8\pi^2\Mp^2)$. For trajectories that reach the fold we record $a_f=e^{x_f}$ and the cosmic time
\begin{equation}
t_f-t_0=\int_0^{x_f}\frac{{\rm d}x}{H(x)},\qquad H=H_0\sqrt{\frac{\Omega_{m0}e^{-3x}+\Ok e^{-2x}}{1-\Omega_d}} .
\label{eq:tfold}
\end{equation}

\subsection{Numerical implementation and validation}
\label{sec:numerics}
Equations~\eqref{eq:q}, \eqref{eq:wd}, \eqref{eq:odot} and~\eqref{eq:udot} are integrated with the explicit Runge--Kutta method DOP853~\cite{Dormand:1980,Hairer:1993} of SciPy~\cite{Virtanen:2019joe}, with relative and absolute tolerances $10^{-7}$ and $10^{-9}$, on 300 points in $0\leq z\leq2.5$. The integration stops, and the point is rejected, if $2-\keff<10^{-4}$ inside the data range (a fold), or if $\Omega_m\to0$ or $\Omega_d\to0$ (for curved GO-HDE, whose solution is the closed form~\eqref{eq:Etilde}, the point is rejected if $\Omega_d\le0$ anywhere in $0\le z\le2.5$; this is the condition $L^2>0$ discussed after Eq.~\eqref{eq:rhodecomp}). Points with $|1-\alpha+\gamma|<10^{-3}$ are rejected. The expansion rate is computed as $E^2=\Omega_{m0}e^{-3x}/\Omega_m$ with $\Omega_m=1-\Omega_d-\Omega_k>0$, which is equivalent to Eq.~\eqref{eq:Erec} but never of the form $0/0$ (in a closed universe $\Omega_d$ may exceed 1). The code passes \NTests{} groups of validation checks (Table~\ref{tab:tests}), after the checks of the flat code~\cite{Pandey:2026flat}: among them, the flat limit (the curved code at $\Ok=0$ equals the flat code of Ref.~\cite{Pandey:2026flat}), the closed form of Proposition~\ref{prop:curved} against numerics for open and closed universes and for both cutoffs, closure of Eq.~\eqref{eq:closure}, the residual of the density relation~\eqref{eq:formA} along the numerical solution, energy conservation, the trackers of Eqs.~\eqref{eq:Odm} and~\eqref{eq:Odk}, Proposition~\ref{prop:first}, Eq.~\eqref{eq:clarkson}, and the fold. A second script, distributed with the code, re-derives every new equation of this paper symbolically with SymPy~\cite{Meurer:2017yhf} and compares each implemented formula with the code that evaluates it; \NEqOK{} of \NEqChecks{} checks pass. To test whether $u_0$ can be recovered, and whether curvature is recovered, we build a noise-free (Asimov~\cite{Cowan:2010js}) mock: every data vector is replaced by the model prediction at a fiducial point with $u_0=\MockUzeroTruth$ and $\Ok=\MockOkTruth$, keeping the real covariances. Its posterior recovers the input values (Sect.~\ref{sec:tension}, Fig.~\ref{fig:mock}).

\begin{table}[H]
\centering\scriptsize
\caption{Validation checks of the code. The result column is copied verbatim from the log of the run that produced this paper.}
\label{tab:tests}
\begin{tabular}{p{0.36\textwidth}cp{0.44\textwidth}}
\toprule
Check & Verdict & Result (from the run log)\\
\midrule
Flat limit: curved code at $\Omega_{k0}=0$ equals the flat code (E(z), $\chi^2$) & pass & \texttt{\detokenize{max |dE|/E = 0.00e+00}}\\
Closed form of curved GO-HDE, Eq.~\eqref{eq:Etilde}, against numerics (GO and Ricci; open and closed) & pass & \texttt{\detokenize{max |dE|/E = 2.56e-10 (6 GO cases, 2 Ricci cases)}}\\
Closure $\Omega_m+\Omega_d+\Omega_k=1$ and $\Omega_k$ from the state against $\Omega_{k0}e^{-2x}/E^2$ & pass & \texttt{\detokenize{max |sum - 1| = 2.2e-16, max |Omega_k state - algebraic| = 1.0e-10}}\\
Constraint residual of the density relation, Eq.~\eqref{eq:formA}, along the numerical solution & pass & \texttt{\detokenize{max |F| = 5.05e-06}}\\
Energy conservation, Eq.~\eqref{eq:wd}, and the equation for $u$, Eq.~\eqref{eq:udot} & pass & \texttt{\detokenize{max |w_num - w_state| = 6.5e-07, max |dln u/dx - 3(1+w)/(2-k)| = 1.5e-05}}\\
Matter-era and curvature-era trackers, Eqs.~\eqref{eq:Odm} and \eqref{eq:Odk} & pass & \texttt{\detokenize{Omega_d -> 0.1000 (alpha - 3 gamma/2 = 0.1000); dln u/dx = 3.000 (3)}}\\
First order in $b$, Eq.~\eqref{eq:firstorder} & pass & \texttt{\detokenize{error 5.08e-04 -> 1.30e-04 when b is halved (expect /4)}}\\
Curvature ratio, Eqs.~\eqref{eq:Etilde} and \eqref{eq:Rkphys}; Ricci cutoff gives $R_k=1$ and equals o$w$CDM & pass & \texttt{\detokenize{R_k = 1.582150 = 3|w|/[(1-Odm)(3|w|-1)]}}\\
Distances (sinh/sin) against quadrature; the Clarkson test is constant & pass & \texttt{\detokenize{max rel. diff = 7.2e-08}}\\
Fold: rejection, location and square-root law & pass & \texttt{\detokenize{ModelError}}\\
Likelihood continuity ($\Omega_{k0}\to0$, $u_0\to0$) and integration tolerance & pass & \texttt{\detokenize{1486.48882 vs 1486.48882}}\\
Parameter maps and the BBN map are unchanged by curvature & pass & \texttt{\detokenize{Odm = alpha - 3 gamma/2, w_GO = -1 - 2(1-alpha)/(3 gamma)}}\\
Pole of the curvature tracker rejected; $w_{\rm GO}$ prior edge & pass & \texttt{\detokenize{|1 - alpha + gamma| < 1e-3 raises ModelError}}\\
BAO distances $D_M$, $D_H$, $D_V$ in curved space against the supernova routine and the definitions & pass & \texttt{\detokenize{max rel. diff D_M 8.8e-09, D_H 0.0e+00, D_V 5.9e-09 (BAO routine vs SN routine and definitions)}}\\
Analytic ${\rm d}w_d/{\rm d}x$, Eq.~\eqref{eq:wprime}, against finite differences & pass & \texttt{\detokenize{analytic 0.39496334, numerical 0.39496337}}\\
Time to the fold, Eq.~\eqref{eq:tfold}, against an independent quadrature & pass & \texttt{\detokenize{t_f - t_0 = 0.34706 Gyr (detector) vs 0.34706 Gyr (trapezoid)}}\\
Nesting $\Lambda$CDM $\subset$ o$\Lambda$CDM $\subset$ o$w$CDM $\subset$ oCPL; Ricci-NTHDE with $u_0\to0$ equals o$w$CDM & pass & \texttt{\detokenize{1540.0786=1540.0786; 1506.0898=1506.0898; 1498.3467=1498.3467}}\\
\bottomrule
\end{tabular}

\end{table}

\section{Results}
\label{sec:results}

\subsection{Constraints and model comparison}
\label{sec:comparison}
Table~\ref{tab:comparison} gives $\chi^2_{\min}$, the information criteria and the Bayes factors of every model against flat $\Lambda$CDM, with both evidence normalizations of Eq.~\eqref{eq:volume}; Table~\ref{tab:posteriors} gives the posterior medians with 68\% intervals of the sampled parameters. Curved GO-NTHDE reaches $\chi^2_{\min}=\ChiGONTHDEK$, i.e.\ $\Delta\chi^2=\DChiGONTHDEK$ relative to flat $\Lambda$CDM, with $\Delta{\rm AIC}=\DAICGONTHDEK$, $\Delta{\rm BIC}=\DBICGONTHDEK$ and $\Delta{\rm DIC}=\DDICGONTHDEK$. Its Bayes factor against flat $\Lambda$CDM is $\ln B=\LnBGONTHDEK$ ($\LnBAGONTHDEK$ over the allowed region). For comparison, flat GO-NTHDE has $\ln B=\LnBGONTHDE$, o$\Lambda$CDM has $\ln B=\LnBOLCDM$, and o$w$CDM $\ln B=\LnBOWCDM$. Three further Bayes factors follow from the same evidences: allowing curvature in GO-NTHDE changes the evidence by $\ln B=\LnBCurvVsFlat$ (curved against flat; $\LnBCurvVsFlatHDE$ for GO-HDE), curved GO-NTHDE against o$\Lambda$CDM gives $\ln B=\LnBvsOLCDMGONTHDEK$, and the Tsallis deformation in the curved model is worth $\ln B=\LnBNTHDEvsHDEK$ (curved GO-NTHDE against curved GO-HDE).

No model is preferred over flat $\Lambda$CDM. Every model with evolving dark energy improves the fit by a similar amount, with $\Delta\chi^2$ from $\DChiWCDM$ (flat $w$CDM) to $\DChiTTWOK$ (curved GO-NTHDE, T2), so the comparison is decided by the number of parameters and the prior volume. The least disfavoured extension is flat $w$CDM by AIC ($\Delta{\rm AIC}=\DAICWCDM$) and flat GO-NTHDE by the evidence ($\ln B=\LnBGONTHDE$). On the Jeffreys scale the evidence against curved GO-NTHDE lies at the boundary between weak and moderate, against flat GO-NTHDE and o$\Lambda$CDM it is weak, and against oCPL strong. Curvature is not rewarded in any model: it lowers the evidence of $\Lambda$CDM by $\LnBOLCDM$ and that of GO-NTHDE by $\LnBCurvVsFlat$, close to the Occam factor of a parameter that the data measure to about $\pm0.04$ inside a prior box $0.6$ wide. The Tsallis deformation, in contrast, is mildly useful in both geometries, as in the flat analysis~\cite{Pandey:2026flat}. The twin check, flat GO-HDE against flat $w$CDM, which share their expansion rate exactly (Proposition~\ref{prop:curved} at $\Ok=0$), gives $\ln B=\LnBTwin$ instead of the exact value zero; we take this as the scale of the sampling error of the evidences, so that differences in $\ln B$ of a few tenths are not significant.

\begin{table}[H]
\centering\scriptsize
\caption{Model comparison on Pantheon+SH0ES + 32 chronometers + DESI DR2 BAO ($N=1702$), relative to flat $\Lambda$CDM. Negative $\Delta\chi^2$, $\Delta$AIC, $\Delta$BIC, $\Delta$DIC and positive $\ln B$ favour the model. $\ln B$ (box) is normalized over the prior box, $\ln B$ (allowed) over the region with a physical solution, Eq.~\eqref{eq:volume}.}
\label{tab:comparison}
\resizebox{\textwidth}{!}{\begin{tabular}{lcccccccc}
\toprule
Model & $k$ & $\chi^2_{\min}$ & $\Delta\chi^2$ & $\Delta$AIC & $\Delta$BIC & $\Delta$DIC & $\ln B$ (box) & $\ln B$ (allowed)\\
\midrule
flat $\Lambda$CDM & \KLCDM & $\ChiLCDM$ & $\DChiLCDM$ & $\DAICLCDM$ & $\DBICLCDM$ & $\DDICLCDM$ & $\LnBLCDM$ & $\LnBALCDM$ \\
flat $w$CDM & \KWCDM & $\ChiWCDM$ & $\DChiWCDM$ & $\DAICWCDM$ & $\DBICWCDM$ & $\DDICWCDM$ & $\LnBWCDM$ & $\LnBAWCDM$ \\
flat GO-HDE & \KGOHDE & $\ChiGOHDE$ & $\DChiGOHDE$ & $\DAICGOHDE$ & $\DBICGOHDE$ & $\DDICGOHDE$ & $\LnBGOHDE$ & $\LnBAGOHDE$ \\
flat GO-NTHDE & \KGONTHDE & $\ChiGONTHDE$ & $\DChiGONTHDE$ & $\DAICGONTHDE$ & $\DBICGONTHDE$ & $\DDICGONTHDE$ & $\LnBGONTHDE$ & $\LnBAGONTHDE$ \\
o$\Lambda$CDM & \KOLCDM & $\ChiOLCDM$ & $\DChiOLCDM$ & $\DAICOLCDM$ & $\DBICOLCDM$ & $\DDICOLCDM$ & $\LnBOLCDM$ & $\LnBAOLCDM$ \\
o$w$CDM & \KOWCDM & $\ChiOWCDM$ & $\DChiOWCDM$ & $\DAICOWCDM$ & $\DBICOWCDM$ & $\DDICOWCDM$ & $\LnBOWCDM$ & $\LnBAOWCDM$ \\
oCPL & \KOCPL & $\ChiOCPL$ & $\DChiOCPL$ & $\DAICOCPL$ & $\DBICOCPL$ & $\DDICOCPL$ & $\LnBOCPL$ & $\LnBAOCPL$ \\
curved GO-HDE & \KGOHDEK & $\ChiGOHDEK$ & $\DChiGOHDEK$ & $\DAICGOHDEK$ & $\DBICGOHDEK$ & $\DDICGOHDEK$ & $\LnBGOHDEK$ & $\LnBAGOHDEK$ \\
curved GO-NTHDE & \KGONTHDEK & $\ChiGONTHDEK$ & $\DChiGONTHDEK$ & $\DAICGONTHDEK$ & $\DBICGONTHDEK$ & $\DDICGONTHDEK$ & $\LnBGONTHDEK$ & $\LnBAGONTHDEK$ \\
curved GO-NTHDE (T2) & \KTTWOK & $\ChiTTWOK$ & $\DChiTTWOK$ & $\DAICTTWOK$ & $\DBICTTWOK$ & $\DDICTTWOK$ & $\LnBTTWOK$ & $\LnBATTWOK$ \\
curved Ricci-NTHDE & \KRICCIK & $\ChiRICCIK$ & $\DChiRICCIK$ & $\DAICRICCIK$ & $\DBICRICCIK$ & $\DDICRICCIK$ & $\LnBRICCIK$ & $\LnBARICCIK$ \\
\bottomrule
\end{tabular}
}
\end{table}

\begin{table}[H]
\centering\small
\caption{Posterior medians and 68\% intervals; $u_0$: 95\% upper limit. $\Omt$ is listed for the GO models and $\Omega_m$ for the others; $w$ for o$w$CDM, $w_0$ for oCPL, $\wgo$ for the GO models.}
\label{tab:posteriors}
\resizebox{\textwidth}{!}{\begin{tabular}{lcccccc}
\toprule
Model & $\tilde\Omega_m$ ($\Omega_m$) & $\Omega_{k0}$ & $w_{\rm GO}$ ($w$) & $u_0$ (95\%) & $H_0$ & $r_d$ [Mpc]\\
\midrule
flat $\Lambda$CDM & $\PostOmLCDM$ & $0$ & $-1$ & $--$ & $\PostHzeroLCDM$ & $\PostRdLCDM$\\
o$\Lambda$CDM & $\PostOmOLCDM$ & $\PostOkOLCDM$ & $-1$ & $--$ & $\PostHzeroOLCDM$ & $\PostRdOLCDM$\\
o$w$CDM & $\PostOmOWCDM$ & $\PostOkOWCDM$ & $\PostWOWCDM$ & $--$ & $\PostHzeroOWCDM$ & $\PostRdOWCDM$\\
oCPL & $\PostOmOCPL$ & $\PostOkOCPL$ & $\PostWzeroOCPL$ & $--$ & $\PostHzeroOCPL$ & $\PostRdOCPL$\\
curved GO-HDE & $\PostOmtGOHDEK$ & $\PostOkGOHDEK$ & $\PostWgoGOHDEK$ & $--$ & $\PostHzeroGOHDEK$ & $\PostRdGOHDEK$\\
curved GO-NTHDE & $\PostOmtGONTHDEK$ & $\PostOkGONTHDEK$ & $\PostWgoGONTHDEK$ & $<\UzeroULGONTHDEK$ & $\PostHzeroGONTHDEK$ & $\PostRdGONTHDEK$\\
curved GO-NTHDE (T2) & $\PostOmtTTWOK$ & $\PostOkTTWOK$ & $\PostWgoTTWOK$ & $<\UzeroULTTWOK$ & $\PostHzeroTTWOK$ & $\PostRdTTWOK$\\
curved Ricci-NTHDE & $\PostOmtRICCIK$ & $\PostOkRICCIK$ & $\PostWgoRICCIK$ & $<\UzeroULRICCIK$ & $\PostHzeroRICCIK$ & $\PostRdRICCIK$\\
\bottomrule
\end{tabular}
}
\end{table}

\subsection{What the data say about curvature}
\label{sec:curvature}
Figure~\ref{fig:okmodels} compares the posterior of $\Ok$ in all curved models, and Fig.~\ref{fig:corner} shows the joint posterior of all parameters of curved GO-NTHDE. Curved GO-NTHDE gives $\Ok=\PostOkGONTHDEK$ (95\% interval $[\OkLoGONTHDEK,\OkHiGONTHDEK]$), with probability $P(\Ok>0)=\POpenGONTHDEK$ of an open universe. The o$\Lambda$CDM value is $\PostOkOLCDM$, o$w$CDM gives $\PostOkOWCDM$ and oCPL $\PostOkOCPL$; the value reported for $w$CDM with curvature from DESI, chronometers and Pantheon+ is $\JHEAPOk$~\cite{Yadav:2026jheap}. Curvature is consistent with zero in every model. The models with evolving dark energy agree closely with one another, with an uncertainty of about $0.04$ around zero, while o$\Lambda$CDM leans to an open universe, $P(\Ok>0)=\POpenOLCDM$. In o$\Lambda$CDM curvature is the only freedom that can imitate the mild evolution of the dark energy that these data prefer; once $w$ is free the preference moves to $w$ ($w=\PostWOWCDM$ in o$w$CDM) and the curvature returns to zero. In curved GO-NTHDE this degeneracy appears as a correlation of $\Ok$ with $\wgo$ (correlation coefficient $\CorrOkWGONTHDEK$), whereas $\Ok$ and $H_0$ are uncorrelated ($\CorrOkHGONTHDEK$) because the SH0ES-calibrated supernovae fix $H_0$. Our constraint is about twice as broad as that of Ref.~\cite{Yadav:2026jheap}, which computes $r_d$ from a BBN prior on the baryon density while we leave $r_d$ free; the two are consistent. The geometric curvature of the GO models is not displaced relative to o$w$CDM, because all of them are consistent with zero; Sect.~\ref{sec:dynamic} discusses what this implies for $\Rk$.

\begin{figure}[H]
\centering\pipefig{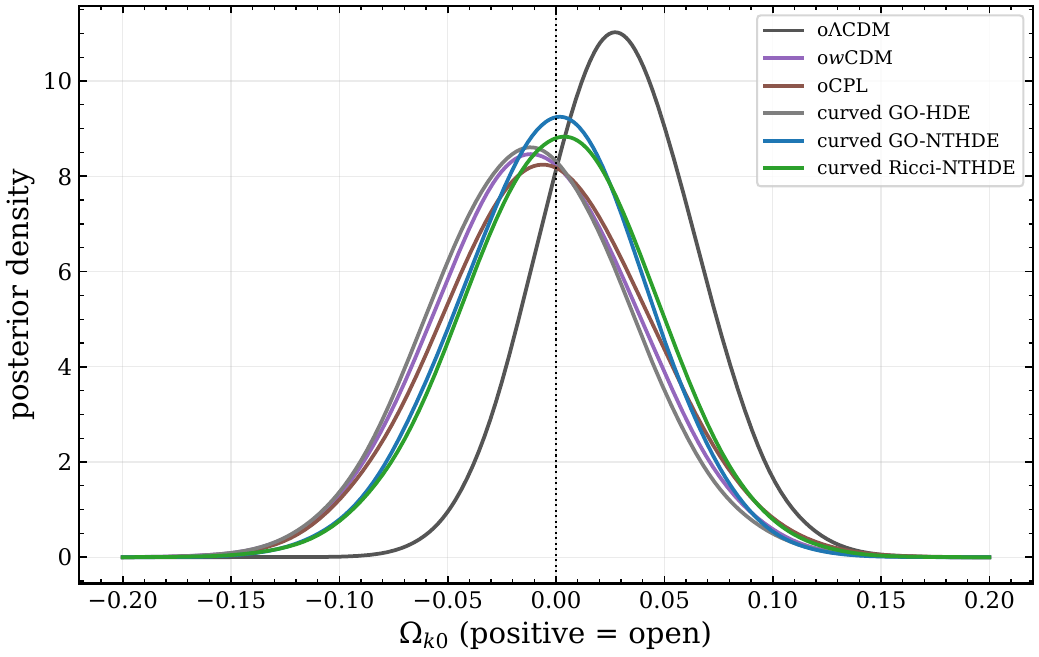}{0.78\textwidth}
\caption{Posterior density of $\Ok$ (positive: open) in the curved models. Dotted: $\Ok=0$.}
\label{fig:okmodels}
\end{figure}

\begin{figure}[H]
\centering\pipefig{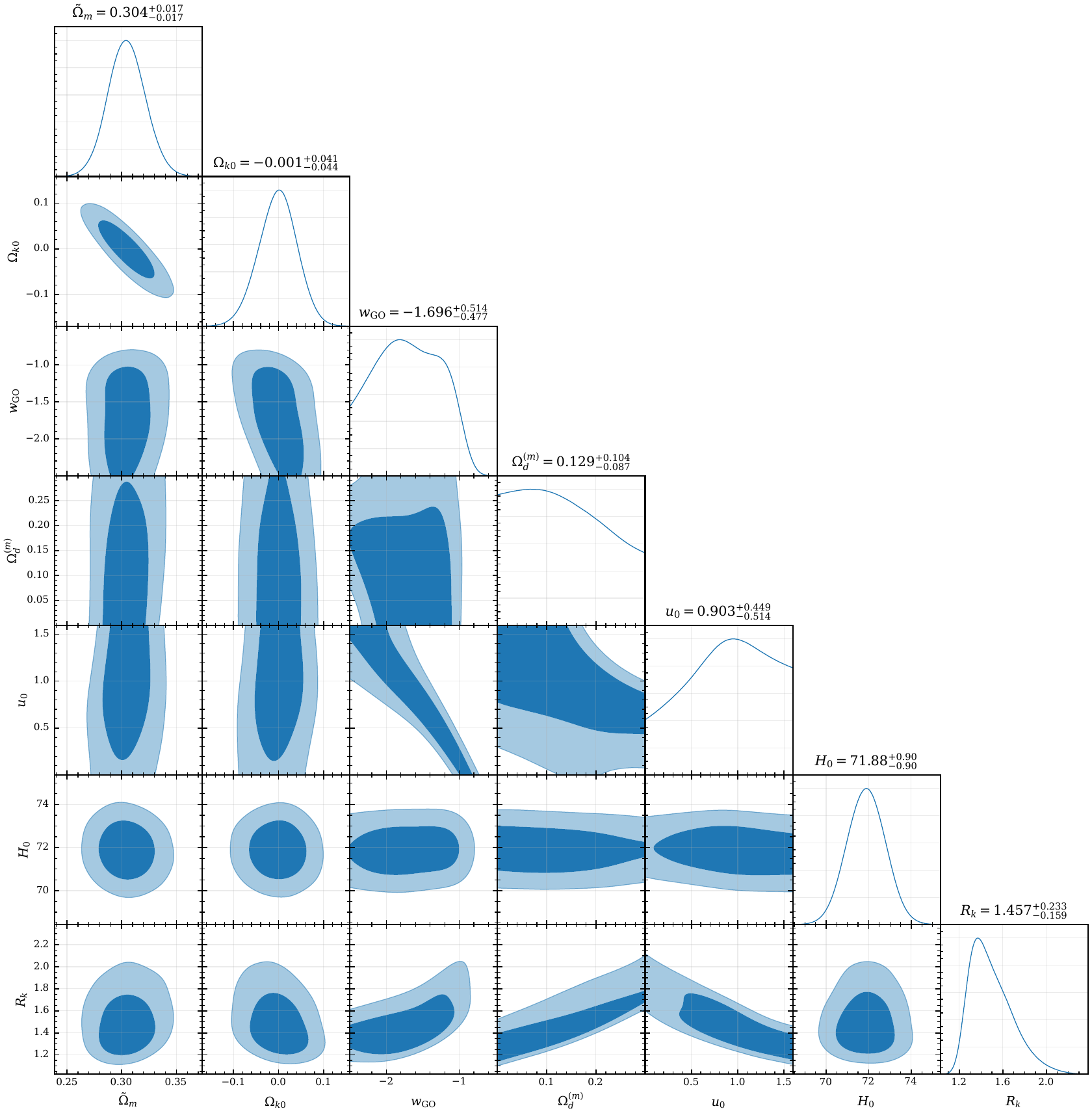}{0.95\textwidth}
\caption{Joint posterior (68\% and 95\%) of curved GO-NTHDE for the sampled parameters and the derived curvature ratio $\Rk$. Titles: medians with 68\% intervals.}
\label{fig:corner}
\end{figure}

\subsection{Geometric versus dynamical curvature}
\label{sec:dynamic}
Proposition~\ref{prop:ratio} predicts that the coefficient of $(1+z)^2$ in $H^2/H_0^2$ is larger than the geometric curvature. In the posterior, $\Rk=\DerRkGONTHDEK$ and $\Okt=\DerOktGONTHDEK$ (Fig.~\ref{fig:dyn}); the matter-era and curvature-era tracker fractions are $\Odm=\PostOdmGONTHDEK$ and $\Odk=\DerOdkGONTHDEK$. For comparison, a curved model in which the two coefficients are tied ($\Rk=1$, o$w$CDM) gives the likelihood in Table~\ref{tab:comparison}; the Ricci-completed GO cutoff, which enforces $\Rk=1$, has $\Delta\chi^2=\DChiRICCIK$ and $\ln B=\LnBRICCIK$, so that $\ln B({\rm GO}:{\rm Ricci})=\LnBGOvsRicci$; curved GO-HDE against o$w$CDM, which differ only in $\Rk$ and in the blind-in-flat-space fraction $\Odm$, gives $\ln B=\LnBGOHDEvsOWCDM$. Both differences are compatible with zero within the sampling error set by the twin check of Sect.~\ref{sec:comparison}: the data cannot yet distinguish $\Rk>1$ from $\Rk=1$. The reason is visible in Fig.~\ref{fig:dyn}: with $\Ok$ this close to zero, the difference $(\Rk-1)\Ok$ between the dynamical and the geometric curvature, Eq.~\eqref{eq:rhodecomp}, is smaller than the uncertainty on $\Okt$ itself. The posterior of $\Rk$ therefore reflects the priors on $\wgo$ and $\Odm$ rather than the data: widening the $\Odm$ box from $[0,0.3]$ to $[0,0.6]$ moves it from $\DerRkGONTHDEK$ to $\RobRkODMWIDE$ (Sect.~\ref{sec:robust}). A measurement of $\Rk$ needs either a curvature clearly different from zero or an independent determination of the geometric curvature, for instance from the Clarkson test, Eq.~\eqref{eq:clarkson}, with future BAO and supernova data, to compare with the coefficient of $(1+z)^2$ in $H^2(z)$.

\begin{figure}[H]
\centering\pipefig{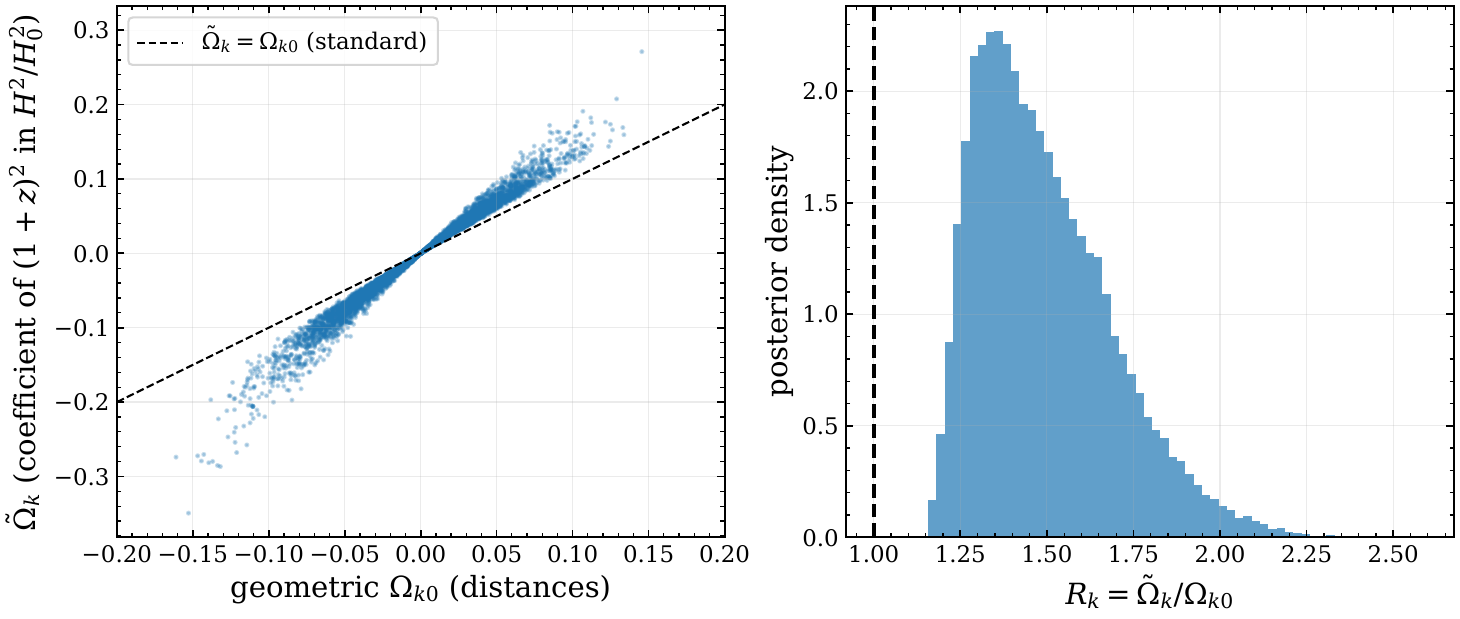}{0.95\textwidth}
\caption{Left: coefficient $\Okt$ of $(1+z)^2$ in $H^2/H_0^2$ against the geometric curvature $\Ok$ for posterior samples of curved GO-NTHDE; the dashed line is the standard relation $\Okt=\Ok$. Right: posterior of $\Rk=\Okt/\Ok$; $\Rk>1$ by Proposition~\ref{prop:ratio}.}
\label{fig:dyn}
\end{figure}

\subsection{Nucleosynthesis with curvature}
\label{sec:bbnres}
Adding the BBN point of Eq.~\eqref{eq:chibbn} gives $\ln B=\BBNLnBGONTHDEK$ for curved GO-NTHDE and $\BBNLnBRICCIK$ for curved Ricci-NTHDE against $\Lambda$CDM+BBN. The distance of the posterior from the Ricci surface is $\Odm-(1+3|\wgo|)^{-1}=\BBNOdmRicci$ (Fig.~\ref{fig:bbn}): with the BBN point the posterior collapses onto the Ricci surface, as in the flat model~\cite{Pandey:2026flat}, which confirms Proposition~\ref{prop:bbn} in a curved universe. On that surface $\Rk$ is a function of $\wgo$ alone, Eq.~\eqref{eq:Rkbbn}, and its posterior, $\Rk=\RobRkBBN$, follows from that of $\wgo$. The lock removes most of the $\Odm$ range and so costs prior volume, which is why curved GO-NTHDE loses evidence when BBN is added, while curved Ricci-NTHDE, which satisfies the bound by construction ($\Odr=0$), does not: with BBN the Ricci completion is preferred over the plain GO cutoff by $\ln B=\BBNLnBRicciMinusGO$.

\begin{figure}[H]
\centering\pipefig{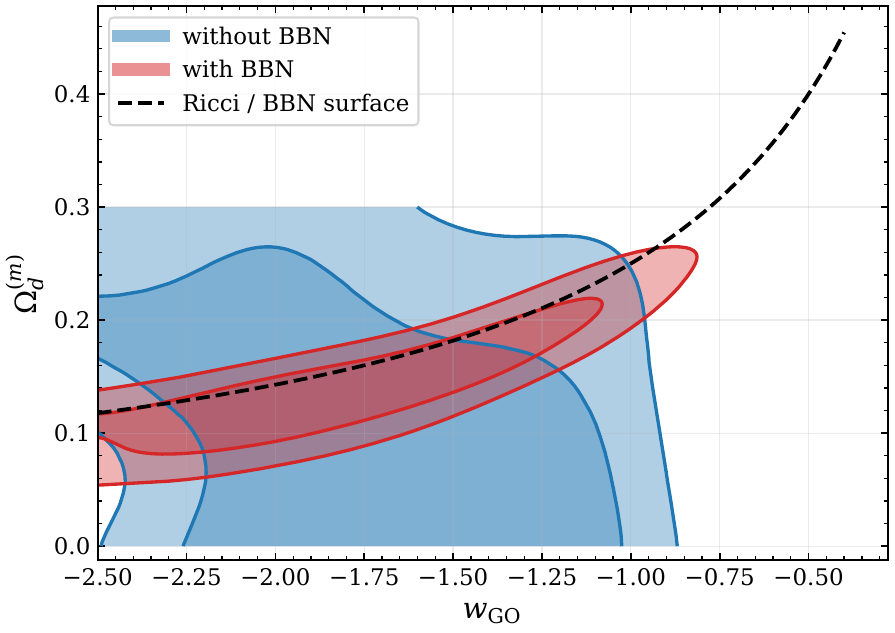}{0.62\textwidth}
\caption{Posterior of curved GO-NTHDE (68\% and 95\% contours from the nested samples) in the $(\wgo,\Odm)$ plane without (blue) and with (red) the BBN data point. Dashed: the Ricci surface $\Odm=1/(1+3|\wgo|)$, Eq.~\eqref{eq:lock}.}
\label{fig:bbn}
\end{figure}

\subsection{Expansion history and the future}
\label{sec:future}
Figure~\ref{fig:bands} shows $w_d(z)$, $q(z)$ and $H/H_{\Lambda{\rm CDM}}$ with 68\% bands. Today $w_0=\DerWzGONTHDEK$, $w_a=\DerWaGONTHDEK$, $q_0=\DerQzGONTHDEK$ and the transition redshift is $z_t=\DerZtGONTHDEK$. A fraction $\PFoldGONTHDEK$ of the posterior reaches the fold of Proposition~\ref{prop:fold}; for those samples the fold is at $a_f=\DerAFoldGONTHDEK$, $t_f-t_0=\DerTFoldGONTHDEK$~Gyr from now (best fit: $a_f=\AFoldBestGONTHDEK$, $\TFoldBestGONTHDEK$~Gyr). The Tsallis parameter is $\delta=\DerDeltaGONTHDEK\times10^{-123}$ (units of Eq.~\eqref{eq:absorb} with $c=1$), and $u_0<\UzeroULGONTHDEK$ at 95\%. The dark energy is of quintessence type today, $P(w_0<-1)=\PPhantomGONTHDEK$, although the GO mode itself is phantom, $\wgo=\PostWgoGONTHDEK$: the Tsallis deformation raises $w_d$ above $\wgo$, and $u_0$ and $\wgo$ are strongly anticorrelated (correlation coefficient $\CorrUWGONTHDEK$). This is how the model follows the mild preference of DESI DR2 with supernovae for $w_0>-1$. The long negative tail of $w_a$ comes from trajectories close to the fold, where $w_d$ changes rapidly. The fold is the price of this behaviour: it is reached by a fraction $\PFoldGONTHDEK$ of the posterior, compared with $\PFoldGONTHDE$ in the flat model, so curvature does not change the future of the model.

\begin{figure}[H]
\centering\pipefig{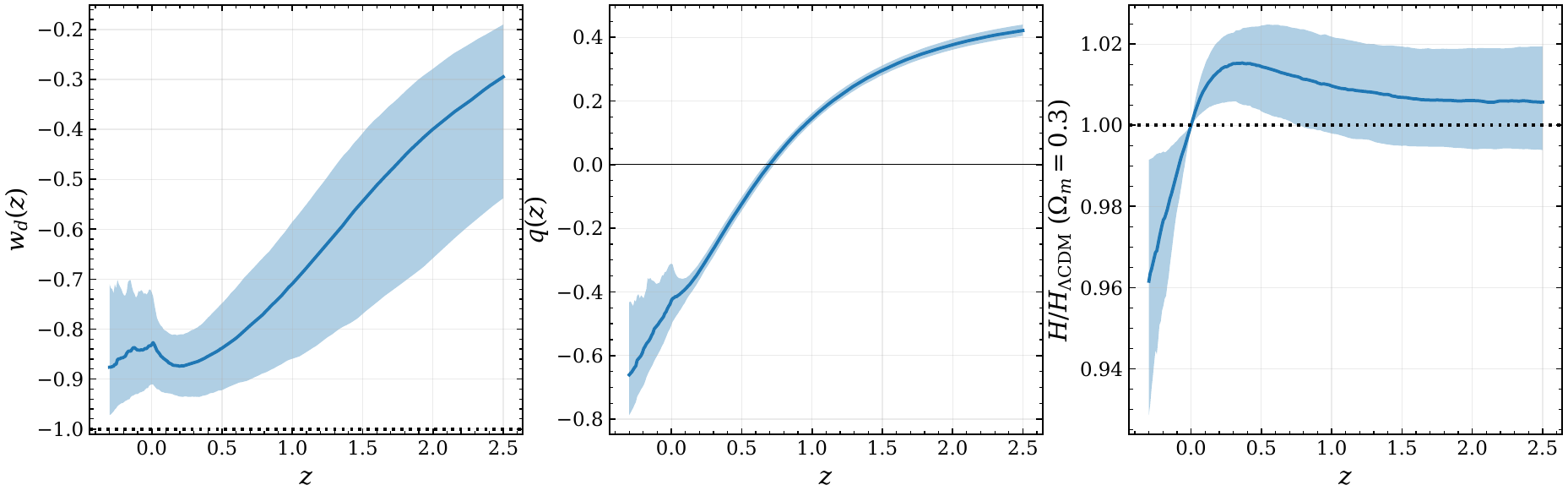}{\textwidth}
\caption{Reconstructed $w_d(z)$ (left), $q(z)$ (middle) and $H/H_{\Lambda{\rm CDM}}$ for $\Omega_m=0.3$ (right) for curved GO-NTHDE; shaded: 68\% from posterior samples; $z<0$ is the future.}
\label{fig:bands}
\end{figure}

\subsection{Consistency of the data sets and mock recovery}
\label{sec:tension}
The tension statistic of Eq.~\eqref{eq:qdmap} between the supernovae and chronometers+BAO is $Q_{\rm DMAP}=\QDMAPOLCDM$ ($\QSigOLCDM\sigma$) for o$\Lambda$CDM and $\QDMAPGONTHDEK$ ($\QSigGONTHDEK\sigma$) for curved GO-NTHDE; $H_0=\HSNGONTHDEK$ (SN) against $\HCBGONTHDEK$ (CC+BAO) (Fig.~\ref{fig:probes}). In both models the two data sets disagree at a similar level, and the disagreement is in $H_0$, between the SH0ES-calibrated supernovae and the chronometers with BAO and a free $r_d$; curvature does not resolve it, because $\Ok$ and $H_0$ are uncorrelated. In the mock test of Sect.~\ref{sec:numerics} (Fig.~\ref{fig:mock}), $\MockNin$ of the $\MockNpar$ input values lie inside the 68\% intervals and all inside the 95\% intervals; the exception is $\Odm$, the parameter the data constrain least (input $\MockOdmTruth$, upper edge of the 68\% interval $\MockOdmHi$). The input curvature $\MockOkTruth$ is recovered as $\Ok=\MockOkPost$, so a curvature of this size would be measured at only about one standard deviation, and $u_0=\MockUzeroTruth$ is recovered with a posterior width $\MockUzeroSd$.

\begin{figure}[H]
\centering\pipefig{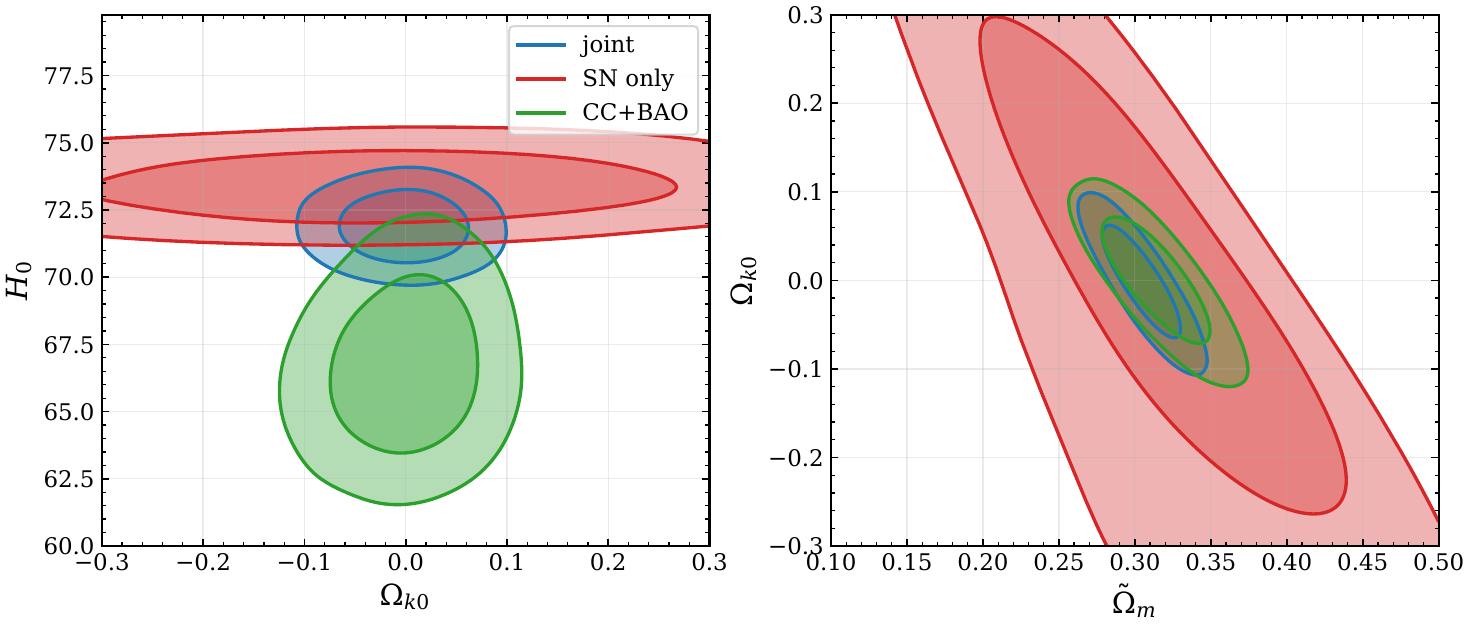}{0.95\textwidth}
\caption{Curved GO-NTHDE: posteriors from supernovae only, chronometers+BAO only and the joint data in the $(\Ok,H_0)$ plane (left) and the $(\Omt,\Ok)$ plane (right).}
\label{fig:probes}
\end{figure}

\begin{figure}[H]
\centering\pipefig{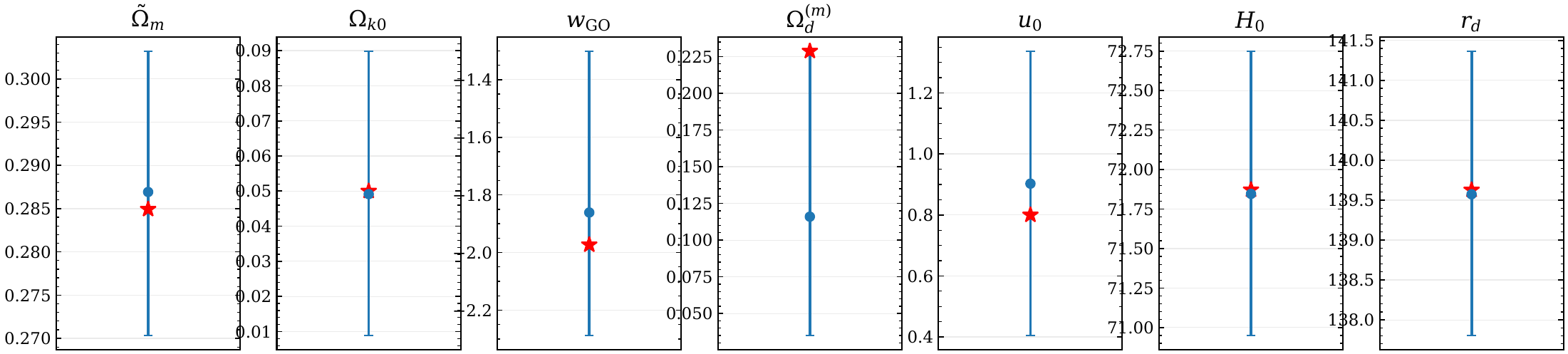}{0.95\textwidth}
\caption{Mock recovery: posterior median and 68\% interval (points) against the input (stars) for each parameter of curved GO-NTHDE.}
\label{fig:mock}
\end{figure}

\subsection{Model comparison across models and robustness}
\label{sec:robust}
Figure~\ref{fig:mcomp} summarizes $\Delta$AIC, $\Delta$BIC and $\ln B$. Table~\ref{tab:robust} and Fig.~\ref{fig:robust} repeat the analysis of curved GO-NTHDE with other supernova compilations, with the curvature prior widened and narrowed, with the prior on $\Odm$ widened to $[0,0.6]$ (because $\Odm$ enters $\Rk$, Eq.~\eqref{eq:Rkphys}, its prior range could shape the posterior of $\Rk$; with the wider box $\Rk=\RobRkODMWIDE$, against $\DerRkGONTHDEK$ in the main run), with the external Planck+BAO curvature term of Eq.~\eqref{eq:chiok}, and with BBN. Three conclusions hold in every variant. First, the curvature is consistent with zero (Fig.~\ref{fig:robust}, right). Second, the evidence depends mainly on the supernova sample: with Union3, curved GO-NTHDE gives $\Delta\chi^2=\RobDChiUNIONTHREE$ and $\ln B=\RobLnBUNIONTHREE$, on a par with $\Lambda$CDM; with DES-Dovekie $\ln B=\RobLnBDOVEKIE$, and with PantheonPlus without SH0ES $\ln B=\RobLnBPPNOSHOES$, in line with the stronger preference of Union3 for evolving dark energy~\cite{DESI:2025zgx}. Third, the curvature prior acts as expected for a parameter measured to about $\pm0.04$: doubling the box lowers $\ln B$ to $\RobLnBOKWIDE$ and narrowing it to $\pm0.1$ raises it to $\RobLnBOKNARROW$, while the Planck+BAO term pins $\Ok=\RobOkOKGAUSS$ and costs the Occam factor estimated in Sect.~\ref{sec:data} ($\ln B=\RobLnBOKGAUSS$). Halving the $u_0$ prior ($\RobLnBUHALF$), widening the $\Odm$ prior ($\RobLnBODMWIDE$) and adding a correlated systematic to the chronometers ($\RobLnBCCSYS$) change $\ln B$ by a few tenths only.

\begin{figure}[H]
\centering\pipefig{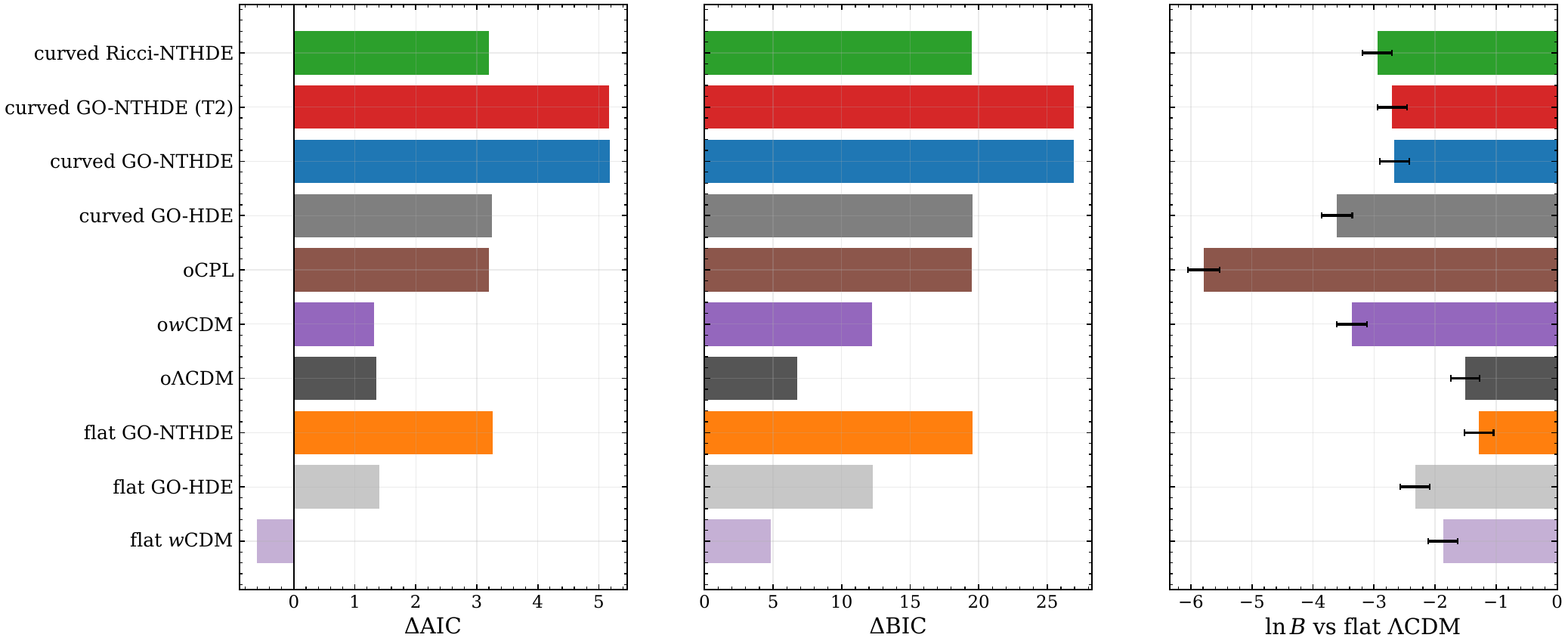}{\textwidth}
\caption{$\Delta$AIC, $\Delta$BIC and $\ln B$ (with error bars) of every model relative to flat $\Lambda$CDM.}
\label{fig:mcomp}
\end{figure}

\begin{table}[H]
\centering\small
\caption{Robustness of curved GO-NTHDE: $\Delta\chi^2$, $\ln B$ against flat $\Lambda$CDM on the same data (both normalizations) and the curvature posterior. Variants add or change one ingredient of the main analysis.}
\label{tab:robust}
\resizebox{\textwidth}{!}{\begin{tabular}{lcccc}
\toprule
Variant & $\Delta\chi^2$ & $\ln B$ (box) & $\ln B$ (allowed) & $\Omega_{k0}$\\
\midrule
Union3 SNe & $\RobDChiUNIONTHREE$ & $\RobLnBUNIONTHREE$ & $\RobLnBAUNIONTHREE$ & $\RobOkUNIONTHREE$ \\
DES-Dovekie SNe & $\RobDChiDOVEKIE$ & $\RobLnBDOVEKIE$ & $\RobLnBADOVEKIE$ & $\RobOkDOVEKIE$ \\
CC with 5\% correlated sys. & $\RobDChiCCSYS$ & $\RobLnBCCSYS$ & $\RobLnBACCSYS$ & $\RobOkCCSYS$ \\
PantheonPlus without SH0ES & $\RobDChiPPNOSHOES$ & $\RobLnBPPNOSHOES$ & $\RobLnBAPPNOSHOES$ & $\RobOkPPNOSHOES$ \\
$\Omega_{k0}$ prior widened to $\pm0.6$ & $\RobDChiOKWIDE$ & $\RobLnBOKWIDE$ & $\RobLnBAOKWIDE$ & $\RobOkOKWIDE$ \\
$\Omega_{k0}$ prior narrowed to $\pm0.1$ & $\RobDChiOKNARROW$ & $\RobLnBOKNARROW$ & $\RobLnBAOKNARROW$ & $\RobOkOKNARROW$ \\
$u_0$ prior halved & $\RobDChiUHALF$ & $\RobLnBUHALF$ & $\RobLnBAUHALF$ & $\RobOkUHALF$ \\
$\Omega_d^{(m)}$ prior widened to $[0,0.6]$ & $\RobDChiODMWIDE$ & $\RobLnBODMWIDE$ & $\RobLnBAODMWIDE$ & $\RobOkODMWIDE$ \\
+ Planck+BAO Gaussian prior on $\Omega_{k0}$ & $\RobDChiOKGAUSS$ & $\RobLnBOKGAUSS$ & $\RobLnBAOKGAUSS$ & $\RobOkOKGAUSS$ \\
+ BBN ($N_\nu$) & $\RobDChiBBN$ & $\RobLnBBBN$ & $\RobLnBABBN$ & $\RobOkBBN$ \\
\bottomrule
\end{tabular}
}
\end{table}

\begin{figure}[H]
\centering\pipefig{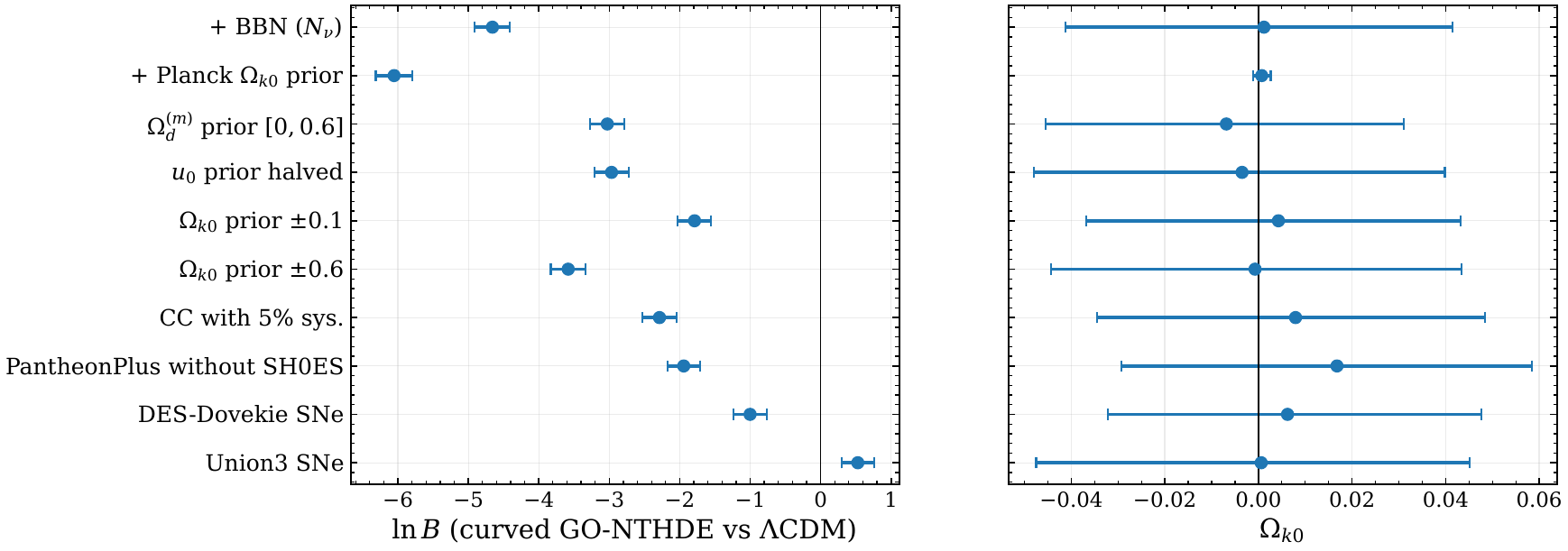}{0.95\textwidth}
\caption{Left: $\ln B$ of curved GO-NTHDE against flat $\Lambda$CDM for each robustness variant. Right: its curvature posterior.}
\label{fig:robust}
\end{figure}

\subsection{Comparison with other holographic models confronted with DESI DR2}
\label{sec:literature}
Table~\ref{tab:literature} collects recent analyses of holographic and generalized-entropy dark energy with DESI DR2 BAO. The supernova samples differ (Pantheon+ with or without the SH0ES calibration, Union3), so the numbers are indicative; for a like-for-like comparison the table also lists our robustness runs with PantheonPlus without SH0ES and with Union3. Two patterns are common to all entries. First, the improvement in $\chi^2$ over $\Lambda$CDM is a few units with PantheonPlus and larger with Union3, as is known for evolving dark energy with DESI DR2~\cite{DESI:2025zgx}; for flat Barrow HDE with the GO cutoff the fitted Barrow exponents are negative, $\Delta=\LitLucPPDelta$ (PantheonPlus) and $\LitLucUDelta$ (Union3)~\cite{Luciano:2026hints}. Second, once the extra parameters are paid for, AIC and the evidence are inconclusive or mildly favour $\Lambda$CDM~\cite{Luciano:2026hints,Luciano:2025jheap,Halder:2026plb}; the one clear preference for a holographic model, $\Delta{\rm AIC}=\LitMahDAIC$ for Barrow HDE with the ``relevant'' GO cutoff of modified Barrow cosmology~\cite{Mahmoudifard:2025pdu}, uses CMB and pre-DESI BAO data. For the curvature, the closest analogue of our analysis is $w$CDM with curvature fitted to DESI DR2, chronometers, BBN and Pantheon+, which gives $\Ok=\LitYadOk$ and $w=\LitYadW$~\cite{Yadav:2026jheap}. Curved GO-NTHDE follows the same pattern. Its $\Delta\chi^2$ values, $\RobDChiPPNOSHOES$ with PantheonPlus without SH0ES and $\RobDChiUNIONTHREE$ with Union3, are close to those of flat Barrow HDE with the GO cutoff on the same supernova samples, $\LitLucPPDChi$ and $\LitLucUDChi$~\cite{Luciano:2026hints}; as there, only with Union3 does the evidence stop disfavouring the model. Its curvature, $\Ok=\PostOkGONTHDEK$, agrees with the $w$CDM value of Ref.~\cite{Yadav:2026jheap}.

\begin{table}[H]
\centering\scriptsize
\caption{Recent confrontations of holographic and generalized-entropy dark energy with DESI DR2 BAO, and this work. $\Delta\chi^2$ and $\Delta$AIC are relative to flat $\Lambda$CDM on the same data; $\ln B$ is the logarithm of the Bayes factor against $\Lambda$CDM where it was computed. Negative $\Delta\chi^2$ and $\Delta$AIC, and positive $\ln B$, favour the model. PP: PantheonPlus without SH0ES; PP+SH0ES: with the Cepheid calibration; U3: Union3; OHD/CC: cosmic chronometers. The entry of Ref.~\cite{Mahmoudifard:2025pdu} uses Planck CMB and lensing with pre-DESI BAO and is listed for completeness.}
\label{tab:literature}
\begin{tabular}{p{0.12\textwidth}p{0.25\textwidth}p{0.20\textwidth}ccc}
\toprule
Reference & Model & Data & $\Delta\chi^2$ & $\Delta$AIC & $\ln B$\\
\midrule
\cite{Luciano:2026hints} & Barrow HDE, GO cutoff, flat & PP + OHD + DESI DR2 & $\LitLucPPDChi$ & $\LitLucPPDAIC$ & --\\
 & & U3 + OHD + DESI DR2 & $\LitLucUDChi$ & $\LitLucUDAIC$ & --\\
\cite{Luciano:2025jheap} & Barrow/Tsallis HDE, event horizon, flat & SN + OHD + DESI DR2 & $\LitLPSDChi$ & $\LitLPSDAIC$ & --\\
\cite{Halder:2026plb} & Viaggiu HDE, event horizon, flat & PP + CC + DESI DR2 & $\LitHalPPDChi$ & $\LitHalPPDAIC$ & $\LitHalPPLnZ$\\
 & & U3 + CC + DESI DR2 & $\LitHalUDChi$ & $\LitHalUDAIC$ & $\LitHalULnZ$\\
\cite{Mahmoudifard:2025pdu} & Barrow HDE, ``relevant'' GO cutoff, flat & CMB + lensing + BAO + CC + SN & $\LitMahDChi$ & $\LitMahDAIC$ & --\\
\midrule
This work & curved GO-NTHDE & PP + SH0ES + CC + DESI DR2 & $\DChiGONTHDEK$ & $\DAICGONTHDEK$ & $\LnBGONTHDEK$\\
 & & PP + CC + DESI DR2 & $\RobDChiPPNOSHOES$ & $\RobDAICPPNOSHOES$ & $\RobLnBPPNOSHOES$\\
 & & U3 + CC + DESI DR2 & $\RobDChiUNIONTHREE$ & $\RobDAICUNIONTHREE$ & $\RobLnBUNIONTHREE$\\
\bottomrule
\end{tabular}
\end{table}

\section{Conclusions}
\label{sec:conclusions}
We have extended new Tsallis holographic dark energy with the Granda--Oliveros cutoff to a curved universe. The background is an autonomous system for $(\Omega_d,u,\Omega_k)$ whose curved GO-HDE limit has the closed-form solution of Karami and Fehri~\cite{Karami:2009gof} (Proposition~\ref{prop:curved}); we have shown that this solution splits the dark energy into a matter-like and a curvature-like part. Four results do not depend on data. First, the dark energy tracks curvature with fraction $\alpha-\gamma$ during curvature domination, in addition to $\alpha-3\gamma/2$ during matter domination. Second, because the GO cutoff lacks the curvature term of the Ricci scalar, $H^2(z)/H_0^2$ contains the coefficient $\Rk\Ok$ of $(1+z)^2$ with $\Rk>1$, while distances and the Clarkson--Bassett--Lu test contain $\Ok$; the Ricci completion $\kappa=\gamma$, $\alpha=2\gamma$ gives $\Rk=1$. Third, the nucleosynthesis lock $\Odm=1/(1+3|\wgo|)$ and the finite-time fold singularity of the flat model survive curvature. Fourth, at first order in the Tsallis deformation the model is GO-HDE plus a constant independent of curvature.

Confronted with Pantheon+SH0ES supernovae, cosmic chronometers and DESI DR2 BAO, the model gives $\Ok=\PostOkGONTHDEK$: late-time data are consistent with a flat universe also when the dark energy is holographic and responds to curvature, and the GO models agree with o$w$CDM and oCPL. Curvature is not rewarded: allowing it lowers the evidence of GO-NTHDE by $\ln B=\LnBCurvVsFlat$, and flat $\Lambda$CDM is preferred over curved GO-NTHDE with weak-to-moderate evidence ($\ln B=\LnBGONTHDEK$ for the model against $\Lambda$CDM), although the model fits slightly better ($\Delta\chi^2=\DChiGONTHDEK$). The verdict depends on the supernova sample: with Union3 the model is on a par with $\Lambda$CDM ($\ln B=\RobLnBUNIONTHREE$). The curvature consistency ratio cannot yet be measured: the GO and Ricci cutoffs fit equally well ($\ln B=\LnBGOvsRicci$), and the posterior of $\Rk$ follows the prior on $\Odm$. Nucleosynthesis locks $\Odm$ onto the Ricci surface with curvature as without it, and with BBN the Ricci completion is preferred ($\ln B=\BBNLnBRicciMinusGO$ for curved Ricci-NTHDE against curved GO-NTHDE). A fraction $\PFoldGONTHDEK$ of the posterior reaches the fold, so the model predicts a generalized sudden singularity in the future unless $u_0$ is small. A measurement of $\Rk$ requires the geometric curvature, from distances, and the coefficient of $(1+z)^2$ in $H^2(z)$ to be determined separately at the $10^{-2}$ level. Limitations: the analysis uses background data only, so curvature is constrained only to a few times $10^{-2}$ and $u_0$ is not identified; CMB distance priors would sharpen $\Ok$ but need the early-time expansion with radiation and a computed $r_d$. A treatment of perturbations, and of other curvature completions of the cutoff beyond the two cases studied here, are natural extensions.

\section*{Data and code availability}
Data are public: Pantheon+SH0ES~\cite{Scolnic:2021amr,Brout:2022vxf}, Union3~\cite{Rubin:2023ovl}, DES-Dovekie~\cite{Popovic:2025dovekie}, cosmic chronometers~\cite{Moresco:2022phi}, DESI DR2~\cite{DESI:2025zgx}. 

\section*{Declaration of competing interest}
The authors declare that they have no known competing financial interests or personal relationships that could have appeared to influence the work reported in this paper.

\bibliographystyle{elsarticle-num}
\bibliography{refs}

\end{document}